\documentclass{theoretics}

\ThCSauthor[ColoradoState]{Nathaniel A. Collins}{naco3124@colostate.edu}[0000-0002-2694-4080]
\ThCSauthor[CSColoradoBoulder,MathColoradoBoulder]{Joshua A. Grochow}{jgrochow@colorado.edu}[0000-0002-6466-0476]
\ThCSauthor[Charleston]{Michael Levet}{levetm@cofc.edu}[0009-0009-1992-3175]
\ThCSauthor[Stuttgart]{Armin Weiß}{armin.weiss@fmi.uni-stuttgart.de}[0000-0002-7645-5867]
\ThCSaffil[ColoradoState]{Department of Mathematics, Colorado State University, Louis R. Weber Building Room 233, 841 Oval Drive, Fort Collins, CO 80521, USA}
\ThCSaffil[CSColoradoBoulder]{Department of Computer Science, University of Colorado Boulder, 1111 Engineering Dr, ECOT 717, 430 UCB, Boulder, CO 80309, USA}
\ThCSaffil[MathColoradoBoulder]{Department of Mathematics, University of Colorado Boulder, Campus Box 395, Boulder, CO 80309 USA}
\ThCSaffil[Charleston]{Department of Computer Science, College of Charleston, 66 George Street, Charleston, SC 29424, USA}
\ThCSaffil[Stuttgart]{Universität Stuttgart\\ Institute for Formal Methods of Computer Science\\Universitätsstraße 38\\
70569 Stuttgart\\
Germany}
\ThCSthanks{A preliminary version of this work appeared in the proceedings of ISSAC 2024 \cite{CGLWISSAC}. JAG and ML were partially supported by JAG's NSF CAREER award CISE-204775 during this work. AW was partially supported by the German Research Foundation (DFG) grant WE 6835/1-2.}
\ThCSshortnames{N.A.\ Collins, J.A.\ Grochow, M.\ Levet, A.\ Weiß}  
\ThCSshorttitle{On the Constant-Depth Circuit Complexity of Generating Quasigroups}
\ThCSyear{2025}
\ThCSarticlenum{19}
\ThCSreceived{Oct 8, 2024}
\ThCSaccepted{Jun 9, 2025}
\ThCSpublished{Aug 22, 2025}
\ThCSdoicreatedtrue
\ThCSkeywords{Group Isomorphism, Minimum Generating Set, Membership, Circuit Complexity, quasiAC$^0$}

\title{On the Constant-Depth Circuit Complexity of Generating Quasigroups}

\usepackage{xspace}
\usepackage{todonotes}
\usepackage{mathtools}
\usepackage[capitalise,noabbrev,nameinlink]{cleveref}
\usepackage{ifthen}
\usepackage{tablefootnote}

\newcommand{\ord}[1]{\operatorname{ord}(#1)}

\declaretheorem[style=italicized,sibling=theorem,name=Definition/Lemma]{deflem}

\newcommand{\jsay}[1]{{\color{red} Josh says: #1}}
\newcommand{\msay}[1]{{\color{blue} Michael says: #1}}

 	\newcommand{\Reach}{\mathop{
 			\mbox{{\fontfamily{cmr}\fontseries{m}\fontshape{n}\selectfont
 					Reach}}}\nolimits}
 	\newcommand{\Accept}{\mathop{
 			\mbox{{\fontfamily{cmr}\fontseries{m}\fontshape{n}\selectfont
 					Accept}}}\nolimits}
 	\newcommand{\Start}{\mathop{
 			\mbox{{\fontfamily{cmr}\fontseries{m}\fontshape{n}\selectfont
 					Start}}}\nolimits}

\newcommand{\algprobm}[1]{\textsc{#1}\xspace}

\newcommand{\betacc}[1]{\ifthenelse{\equal{#1}{1}}{\exists^{\log n}}{\exists^{\log^{#1}n}}} 
\newcommand{\alphacc}[1]{\ifthenelse{\equal{#1}{1}}{\forall^{\log n}}{\forall^{\log^{#1}n}}} 

\newcommand{\Oh}{O}
\newcommand{\Soc}{\text{Soc}}

\newcommand{\Z}{\mathbb{Z}}
\newcommand{\N}{\mathbb{N}}
\newcommand{\abs}[1]{\left|\mathinner{#1}\right|}
\newcommand{\ceil}[1]{\left\lceil\mathinner{#1}\right\rceil}

\renewcommand{\setminus}{\mysetminus}
\newcommand{\mysetminusD}{\raisebox{.8pt}{\hbox{\tikz{\draw[line width=0.6pt,line cap=round] (3.5pt,0pt) -- (0,5.2pt);}}}}
\newcommand{\mysetminusT}{\mysetminusD}
\newcommand{\mysetminusS}{\raisebox{.5pt}{\hbox{\tikz{\draw[line width=0.45pt,line cap=round] (2.2pt,0) -- (0,3.8pt);}}}}
\newcommand{\mysetminusSS}{\raisebox{.35pt}{\hbox{\tikz{\draw[line width=0.4pt,line cap=round] (1.5pt,0) -- (0,2.8pt);}}}}

\newcommand{\mysetminus}{\mathbin{\mathchoice{\mysetminusD}{\mysetminusT}{\mysetminusS}{\mysetminusSS}}}

\ThCSnewtheostd{question}
\crefname{observation}{Observation}{Observations}

\newcommand{\Lem}[1]{Lem.~\ref{#1}\xspace}

\newcommand{\Cor}[1]{Cor.~\ref{#1}\xspace}
\newcommand{\Prop}[1]{Prop.~\ref{#1}\xspace}
\newcommand{\Thm}[1]{Thm.~\ref{#1}\xspace}

\DeclareMathOperator{\poly}{poly}

\newcommand*{\ComplexityClass}[1]{\ensuremath{\mathsf{#1}}\xspace}

\newcommand*{\NL}{\ComplexityClass{NL}}
\newcommand*{\NP}{\ComplexityClass{NP}}
\renewcommand{\P}{\ComplexityClass{P}}

\newcommand*{\EXP}{\ComplexityClass{EXP}}

\newcommand{\NTISP}{\ComplexityClass{NTISP}}

\newcommand{\DTISP}{\ComplexityClass{DTISP}}

\newcommand{\NTISPpll}{\ensuremath{\NTISP(\mathrm{polylog}(n),\mathrm{log}(n))}\xspace}
\newcommand{\DTISPpll}{\ensuremath{\DTISP(\mathrm{polylog}(n),\mathrm{log}(n))}\xspace}

\newcommand*{\LogSpace}{\ComplexityClass{L}}

\newcommand*{\PSPACE}{\ComplexityClass{PSPACE}}

\newcommand{\AC}{\ComplexityClass{AC}}
\newcommand{\SAC}{\ComplexityClass{SAC}}
\newcommand{\MACz}{\ComplexityClass{MAC^0}}

\newcommand{\ACz}{\ComplexityClass{AC^0}}
\newcommand{\TCz}{\ComplexityClass{TC^0}}
\newcommand{\qACz}{\ComplexityClass{quasiAC^0}}
\newcommand{\FOLL}{\ComplexityClass{FOLL}}

\newcommand{\DLOGTIME}{\ComplexityClass{DLOGTIME}} %
\newcommand{\LOGSPACE}{\ComplexityClass{L}} %

\newcommand{\DSPACE}{\ComplexityClass{DSPACE}}

\newcommand{\DTIME}{\ComplexityClass{DTIME}}
\newcommand{\NTIME}{\ComplexityClass{NTIME}}

\begin{document}
\maketitle

\begin{abstract}
We investigate the constant-depth 
circuit complexity 
of the \algprobm{Isomorphism Problem}, \algprobm{Minimum Generating Set Problem} (\algprobm{MGS}), and \algprobm{Sub(quasi)group Membership Problem} (\algprobm{Membership})  for groups and quasigroups (=Latin squares), given as input in terms of their multiplication (Cayley) tables.
Despite decades of research on these problems, lower bounds for these problems even against  depth-$2$ $\AC$ circuits remain unknown. Perhaps surprisingly, Chattopadhyay, Torán, and Wagner (FSTTCS 2010; \emph{ACM Trans. Comput. Theory}, 2013) showed that \algprobm{Quasigroup Isomorphism} could be solved by $\AC$ circuits of depth $O(\log \log n)$ using $O(\log^2 n)$ nondeterministic bits, a class we denote $\betacc{2}\FOLL$. We narrow this gap by improving the upper bound for many of these problems to $\qACz$, thus decreasing the depth to constant.

In particular, we show that \algprobm{Membership} can be solved in $\NTIME(\mathrm{polylog}(n))$ and use this to prove the following:
\begin{itemize}
\item \algprobm{MGS} for quasigroups belongs to $\betacc{2}\alphacc{1}\NTIME(\mathrm{polylog}(n))$ $\subseteq \qACz$. Papadimitriou and Yannakakis (\emph{J. Comput. Syst. Sci.}, 1996) conjectured that this problem was $\betacc{2}\P$-complete; our results refute a version of that conjecture for completeness under $\qACz$ reductions unconditionally, and under polylog-space reductions assuming $\EXP \neq \PSPACE$. 

It furthermore implies that this problem is not hard for any class containing \algprobm{Parity}. The analogous results concerning \algprobm{Parity} were known for \algprobm{Quasigroup Isomorphism} (Chattopadhyay, Tor\'an, \& Wagner, \emph{ibid.}) and \algprobm{Subgroup Membership} (Fleischer, \textit{Theory Comput.} 2022), though not for \algprobm{MGS}.

\item \algprobm{MGS} for groups belongs to $\AC^{1}(\LogSpace)$. Our $\AC^{1} (\LogSpace)$ bound improves on the previous, very recent, upper bound of $\mathsf{P}$ (Lucchini \& Thakkar, \textit{J. Algebra}, 2024). Our $\qACz$ upper bound is incomparable to $\P$, but has similar consequences to the above result for quasigroups. 

\item \algprobm{Quasigroup Isomorphism} belongs to $\betacc{2}\ACz(\mathsf{TIMESPACE}(\text{polylog}(n), \log(n))) \subseteq \qACz$. As a consequence of this result and previously known $\ACz$ reductions, this implies the same upper bound for the \algprobm{Isomorphism Problems} for: Steiner triple systems, pseudo-STS graphs, \algprobm{Latin Square Isotopy}, Latin square graphs, and Steiner $(t,t+1)$-designs. 
This improves upon the previous upper bound for these problems, which was $\betacc{2}\LogSpace \cap \betacc{2}\FOLL \subseteq \mathsf{quasiFOLL}$ (Chattopadhyay, Tor\'an, \& Wagner, \emph{ibid.}; Levet, \textit{Australas. J. Combin.} 2023).

\item  As a strong contrast, we show that MGS for arbitrary magmas is \NP-complete.
\end{itemize}

Our results suggest that understanding the constant-depth circuit complexity may be key to resolving the complexity of problems concerning (quasi)groups in the multiplication table model.
\end{abstract}



\setcounter{page}{1}

\section{Introduction}
The \algprobm{Group Isomorphism} (\algprobm{GpI}) problem is a central problem in computational complexity and computer algebra. When the groups are given as input by their multiplication (a.k.a. Cayley) tables, the problem reduces to \algprobm{Graph Isomorphism} (\algprobm{GI}), and because the best-known runtimes for the two are quite close ($n^{O(\log n)}$ \cite{MillerTarjan}\footnote{Miller \cite{MillerTarjan} credits Tarjan for $n^{\log n + O(1)}$.} vs. $n^{O(\log^2 n)}$ \cite{BabaiGraphIso}\footnote{Babai \cite{BabaiGraphIso} proved quasi-polynomial time, and the exponent of the exponent was analyzed and improved by Helfgott~\cite{HelgottGIAnalysis}.}), the former stands as a key bottleneck towards further improvements in the latter. 

Despite this, \algprobm{GpI} seems quite a bit easier than \algprobm{GI}. For example, Tarjan's $n^{\log n + O(1)}$ algorithm for groups \cite{MillerTarjan} can now be given as an exercise to undergraduates: every group is generated by at most $\lfloor \log_2 |G| \rfloor$ elements, so the algorithm is to try all possible $\binom{n}{\log n} \leq n^{\log n}$ generating sets, and for each, check in $n^{O(1)}$ time whether the map of generating sets extends to an isomorphism. In contrast, the quasi-polynomial time algorithm for graphs was a tour de force that built on decades of research into algorithms and the structure of permutation groups. Nonetheless, it remains unknown whether the problem for groups is actually easier than that for graphs under polynomial-time reductions, or even whether both problems are in $\P$!

Using a finer notion of reduction, Chattopadhyay, Torán, and Wagner \cite{ChattopadhyayToranWagner} proved that there was no $\ACz$ reduction from \algprobm{GI} to \algprobm{GpI}. This gave the first unconditional evidence that there is \emph{some formal sense} (namely, the $\ACz$ sense) in which \algprobm{GpI} really is easier than \algprobm{GI}. The key to their result was that the generator-enumeration technique described above can be implemented by non-deterministically guessing $\log^2 n$ bits (describing the $\log n$ generators, each of $\log n$ bits), and then verifying an isomorphism by a non-deterministic circuit of depth only $O(\log \log n)$, 
a class we denote $\betacc{2}\FOLL$. Observe that $\betacc{2}\FOLL \subseteq \mathsf{quasiFOLL}$ (by trying all $2^{O(\log^2 n)}$ settings of the non-deterministic bits in parallel), which cannot compute \algprobm{Parity} \cite{Hastad}, even if augmented with $\mathsf{Mod}_p$ gates for $p$ an odd prime \cite{Razborov, Smolensky87algebraicmethods, ChattopadhyayToranWagner}.  As \algprobm{GI} is $\mathsf{DET}$-hard \cite{Toran}---and hence can compute \algprobm{Parity}---there can be no $\ACz$ reduction from \algprobm{GI} to \algprobm{GpI}.

Such a low-depth circuit was quite surprising, although that surprise is perhaps tempered by the use of non-determinism. Nonetheless, it raises the question: 

\begin{quotation}
\noindent Is it possible that \algprobm{Group Isomorphism} is in $\mathsf{AC}^0$?
\end{quotation}

\noindent The authors would be shocked if the answer were ``yes,'' and yet we do not even have results showing that \algprobm{Group Isomorphism} cannot be computed by polynomial-size circuits of (!) depth 2. Indeed, it is not clear how to use existing $\ACz$ lower bound techniques against \algprobm{Group Isomorphism}.\footnote{The $\betacc{2}\FOLL$ upper bound unconditionally rules out reductions from \algprobm{Parity} and \algprobm{Majority}. While switching lemmas have been used to get $\ACz$ lower bounds on \algprobm{LogClique} (deciding if a graph has a clique on $O(\log n)$ vertices) \cite{lynch, Beame1990, RossmanClique} (covered in Beame's switching lemma primer \cite{BeamePrimer}), which is in $\betacc{2}\ACz \subseteq \qACz$, that problem feels quite different from \algprobm{Group Isomorphism}.}

In this paper, we aim to close the gap between $\ACz$ and $\betacc{2}\FOLL$ in the complexity of \algprobm{Group Isomorphism} and related problems. 
Our goal is to obtain constant-depth circuits of quasipolynomial size, a natural benchmark in circuit complexity \cite{BarringtonQuasipolynomial}. Significantly improving the size of these circuits would improve the state of the art run-time of \algprobm{Group Isomorphism}, a long-standing open question that we do not address here.

Our first main result along these lines is:

\begin{theorem}
\algprobm{(Quasi)Group Isomorphism} can be solved in $\qACz$.
\end{theorem}

(We discuss quasigroups more below.) 

\begin{remark}
We in fact get a more precise bound of $\betacc{2}\ACz(\DTISPpll)$ (Theorem~\ref{thm:QuasigroupIsoPolylogtime} gives an even more specific bound), where $\mathsf{DTISP}(t(n),s(n))$ is the class of languages decidable by a Turing machine that simultaneously uses time at most $t$ and space at most $s$. This more precise bound is notable because it is contained in both $\qACz$ and $\betacc{2}\FOLL \cap \betacc{2} \LogSpace$, the latter thus improving on \cite{ChattopadhyayToranWagner}. We get similarly precise bounds with complicated-looking complexity classes for the other problems mentioned in the introduction, but we omit the precise bounds here for readability.
\end{remark}

  Focusing on depth bounds---that is, without attempting to improve the worst-case runtime---our result is close to the end of the line for a series of works stretching back to 1970; see Table~\ref{table:history}. The only possible further improvements we see, without improving the worst-case runtime, are to get an $\betacc{2}\ACz$ upper bound---for which there are several obstacles, see Section~\ref{sec:conclusion}---or improving the exact size and depth of our result, but there is not much room for improvement here, as we already get quasi-polynomial-size circuits of depth only 4 and size $n^{\Oh(\log n)}$, matching the current-best serial runtime up to the constant in the exponent (we have not attempted to optimize the constant hidden in the big-Oh; we certainly do not get a constant less than 1, and we conservatively estimate our proof yields a constant not more than $20$).

\begin{table}[!htbp]
\small
\begin{center}
\begin{tabular}{|r|c|c|l|}
\hline
\textbf{Year} & \textbf{Result} & \textbf{Depth} & \textbf{Citation} \\ \hline
1970 & Generator-enumerator introduced & $\poly(n)$ & Felsch \& Neub\"{u}ser \cite{FN}\tablefootnote{Complexity not analyzed there, but the same as Tarjan's algorithm \cite{MillerTarjan}.} \\ \hline
1978 & $\betacc{2} \P \subseteq \mathsf{DTIME}(n^{\log n + O(1)})$ & $\poly(n)$ & Tarjan (see Miller \cite{MillerTarjan})  \\ \hline
1977 & $\mathsf{DSPACE}(\log ^2 n)$ & $O(\log^2 n)$ & Lipton--Snyder--Zalcstein \cite{LiptonSnyderZalcstein}\tablefootnote{Despite the publication dates, this seems to have been independent of Tarjan's result. They note that Miller and Rabin had also observed this result independently.}  \\ \hline
 1994 & $\betacc{2}\mathsf{AC}^1$ & $O(\log n)$ & Wolf \cite{Wolf}\tablefootnote{\label{fn:wolf}Wolf only claims a bound of $\betacc{2}\mathsf{NC}^{2}$. However, if we replace his use of $\mathsf{NC}^{1}$ circuits to multiply two elements of a quasigroup with $\ACz$ circuits, we immediately get the $\betacc{2}\mathsf{AC}^1$ bound.} \\ \hline
 2010 & $\betacc{2}\mathsf{SAC}^1$ & $O(\log n)$ & Wagner \cite{WagnerThesis} \\ \hline
 2010 & $\betacc{2}\FOLL \cap \betacc{2} \LogSpace$ & $O(\log \log n)$ & Chattopadhyay--Torán--Wagner \cite{ChattopadhyayToranWagner}\tablefootnote{They do not claim the $\betacc{2}\LogSpace$ bound, but it follows immediately from their algorithm and results.} \\ \hline
 2013 & $\betacc{2}\mathsf{SC}^2 \cap \betacc{2}\LogSpace$ & $O(\log n)$ & Papakonstantinou--Tang--Qiao \cite{TangThesis}\tablefootnote{We have written the result this way, despite $\betacc{2} \LogSpace \subseteq \betacc{2}\mathsf{SC}^2$, because in Tang's thesis \cite{TangThesis}, the only place this is currently published, they only claim $\mathsf{NSC}^2$ using only $O(\log^2 n)$ bits of nondeterminism, which in our notation would be $\betacc{2}\mathsf{SC}^2$. However, their algorithm and results also immediately yields a $\betacc{2}\LogSpace$ bound.}  
 \\ \hline
 2024 &\parbox{3.1in}{\centering $\betacc{2}\ACz(\DTISPpll)$ \\ $ \subseteq \qACz$} & 4 & This work \\ \hline
\end{tabular}
\end{center}

\caption{\label{table:history} History of the low-level circuit complexity of algorithms for \algprobm{(Quasi)Group Isomorphism} based on the generator-enumerator technique. For non-circuit classes, we list their depth as the best-known depth of their simulation by circuits. The class in our bound is contained in all the other classes listed in the table. Although depth 4 in our result does not follow from the complexity class as listed here, it follows from the more exact class we use in Theorem~\ref{thm:QuasigroupIsoPolylogtime}. }
\end{table}

We note that although there are depth reduction techniques for bounded-depth circuits with $\mathsf{Mod}_p$ gates  \cite{AllenderH94}, no such result is known for $\qACz$ circuits (without $\mathsf{Mod}_p$ gates).

\paragraph{Minimum generating set.} Another very natural problem in computational algebra is the \algprobm{Min Generating Set} (\algprobm{MGS}) problem. Given a group, this problem asks to find a generating set of the smallest possible size. Given that many algorithms on groups depend on the size of a generating set, finding a minimum generating set has the potential to be a widely applicable subroutine. The \algprobm{MGS} problem for groups was shown to be in $\P$ by Lucchini \& Thakkar only very recently \cite{LucchiniThakkar}. We improve their complexity bound to:

\begin{theorem}
\algprobm{Min Generating Set} for groups can be solved in $\qACz$ and in $\mathsf{AC}^1(\LogSpace)$ ($O(\log n)$-depth, unbounded fan-in circuits with a logspace oracle).
\end{theorem}

\noindent We note that, although $\qACz$ is incomparable to $\P$ because of the quasi-polynomial size (whereas $\AC^1(\LogSpace) \subseteq \P$), the key we are focusing on here is reducing the depth. 

For nilpotent groups (widely believed to be the hardest cases of \algprobm{GpI}), if we only wish to compute the minimum \emph{number} of generators, we can further improve this complexity to $\LogSpace \cap \ACz(\NTISPpll)$ (Proposition~\ref{prop:Nilpotent}).

While our $\AC^1(\LogSpace)$ bound above is essentially a careful complexity analysis of the polynomial-time algorithm of Lucchini \& Thakkar \cite{LucchiniThakkar}, the $\qACz$ upper bound is in fact a consequence of our next, more general, result for \emph{quasi}groups, which involves some new ingredients.

\paragraph{Enter quasigroups.} Quasigroups can be defined in (at least) two equivalent ways: (1) an algebra whose multiplication table is a Latin square,\footnote{A Latin square is an $n \times n$ matrix where for each row and each column, the elements of $[n]$ appear exactly once.} or (2) a group-like algebra that need not have an identity nor be associative, but in which left and right division are uniquely defined, that is, for all $a,b$, there are unique $x$ and $y$ such that $ax=b$ and $ya=b$. 

In the paper in which they introduced $\log^2(n)$-bounded nondeterminism, Papadimitriou and Yannakakis showed that for arbitrary magmas,\footnote{A magma is a set $M$ together with a function $M \times M \to M$ that need not satisfy any additional axioms.} testing whether the magma has $\log n$ generators was in fact \emph{complete} for $\betacc{2} \P$, and conjectured:

\begin{conjecture}[{Papadimitriou \& Yannakakis \cite[p.~169]{PY}}] \label{conj:PY}
\algprobm{Min Generating Set for Quasigroups} is $\betacc{2}\P$-complete.
\end{conjecture}

\noindent They explicitly did \emph{not} conjecture the same for \algprobm{MGS} for \emph{groups}, writing:

\begin{quotation}
\noindent ``We conjecture that this result [$\betacc{2}\P$-completeness] also holds for the more structured MINIMUM GENERATOR SET OF A QUASIGROUP problem. In contrast, QUASIGROUP ISOMORPHISM was recently shown to be in DSPACE$(\log^2 n)$ \cite{Wolf}. Notice that the corresponding problems for groups were known to be in DSPACE$(\log^2 n)$ \cite{LiptonSnyderZalcstein}.''---Papadimitriou \& Yannakakis \cite[p.~169]{PY}
\end{quotation}

\noindent We thus turn our attention to the analogous problems for quasigroups: \algprobm{MGS} for quasigroups, \algprobm{Quasigroup Isomorphism}, and the key subroutine, \algprobm{Sub-quasigroup Membership}. We note that the $\betacc{2}\FOLL$ upper bound of Chattopadhyay, Torán, and Wagner \cite{ChattopadhyayToranWagner} actually applies to \algprobm{Quasigroup Isomorphism} and not just \algprobm{GpI}; we perform a careful analysis of their algorithm to put \algprobm{Quasigroup Isomorphism} into $\qACz$ as well. 

\begin{theorem} \label{thm:main-mgsq}
\algprobm{Min Generating Set for Quasigroups} is in $\qACz \cap \DSPACE(\log^2 n)$. 
\end{theorem}

To the best of our knowledge, \algprobm{MGS for Quasigroups} has not been studied from the complexity-theoretic viewpoint previously. While a $\DSPACE(\log^2 n)$ upper bound for $\algprobm{MGS}$ for \emph{groups} follows from \cite{TangThesis, ArvindToran}, as far as we know it remained open for quasigroups prior to our work. 

As with prior results on \algprobm{Quasigroup Isomorphism} and \algprobm{Group Isomorphism} \cite{ChattopadhyayToranWagner}, and other isomorphism problems such as  \algprobm{Latin Square Isotopy} and  \algprobm{Latin Square Graph Isomorphism} \cite{LevetLatinSquares}, Thm.~\ref{thm:main-mgsq} implies that \algprobm{Parity} does not reduce to \algprobm{MGS for Quasigroups}, thus ruling out most known lower bound methods that might be used to prove that \algprobm{MGS for Quasigroups} is not in $\ACz$. We also observe a similar bound for \algprobm{MGS for Groups} using Fleischer's technique \cite{Fleischer}. 

Papadimitriou and Yannakakis did not specify the type of reduction used in their conjecture, though their $\betacc{2}\P$-completeness result for \algprobm{Log Generating Set of a Magma} works in both logspace and $\ACz$ (under a suitable input encoding). Our two upper bounds rule out such reductions for \algprobm{MGS for Quasigroups} (unconditionally in one case, conditionally in the other):

\begin{corollary}
The conjecture of Papadimitriou \& Yannakakis \cite[p.~169]{PY} is false under $\qACz$ reductions. It is also false under polylog-space reductions assuming $\EXP \neq \PSPACE$.
\end{corollary}

In strong contrast, we show that \algprobm{MGS for Magmas} is \NP-complete (Thm.~\ref{thm:magma}).

A key ingredient in our proof of Thm.~\ref{thm:main-mgsq} is an improvement in the complexity of another central problem in computational algebra: the \algprobm{Sub-quasigroup Membership} problem (\algprobm{Membership},\footnote{In the literature, the analogous problem for groups is sometimes called \algprobm{Cayley Group Membership} or \algprobm{CGM}, to highlight that it is in the Cayley table model.} for short):

\begin{theorem} \label{thm:main-cqm}
\algprobm{Membership} for quasigroups is in $\betacc{2}\DTISPpll \subseteq \qACz$.
\end{theorem}

 \algprobm{Membership} for \emph{groups} is well-known to belong to $\LogSpace$, by reducing to the connectivity problem on the appropriate Cayley graph (cf. \cite{BarringtonMcKenzie, Reingold}), but as $\LogSpace$ sits in between $\ACz$ and $\AC^1$, this is not low enough depth for us.

\paragraph{Additional results.} We also obtain a number of additional new results on related problems, some of which we highlight here:
\begin{itemize}
\item 
By known $\ACz$ reductions (see, e.g., Levet \cite{LevetLatinSquares} for details), our $\qACz$ analysis of Chattopadhyay, Torán, and Wagner's algorithm for \algprobm{Quasigroup Isomorphism} yields $\qACz$ upper bound for the isomorphism problems for Steiner triple systems, pseudo-STS graphs, Latin square graphs, and Steiner $(t,t+1)$-designs, as well as \algprobm{Latin Square Isotopy}. Prior to our work, \algprobm{Quasigroup Isomorphism} was not known to be solvable using $\mathsf{quasiAC}$ circuits of depth $o(\log \log n)$. See Cor.~\ref{cor:SRGs}.

\item \algprobm{Group Isomorphism} for simple groups (Cor.~\ref{cor:simple}) or for groups from a dense set $\Upsilon$ of orders (Thm.~\ref{thm:ParallelDW}) can be solved in $\ACz(\DTISPpll) \subseteq \qACz \cap \LogSpace \cap \FOLL$. For groups in a dense set of orders, this improves the parallel complexity compared to the original result of Dietrich \& Wilson \cite{DietrichWilson}. As in their paper, note that $\Upsilon$ omits large prime powers. Thus, we essentially have that for groups that are \textit{not} $p$-groups, \algprobm{Group Isomorphism} belongs to a \textit{proper} subclass of $\mathsf{DET}$. This evidence fits with the widely-believed idea that $p$-groups are a bottleneck case for \algprobm{Group Isomorphism}.

\item \algprobm{Abelian Group Isomorphism} (Thm.~\ref{thm:abelian}) is in $\forall^{\log \log n} \mathsf{MAC}^0(\DTISPpll)$. The key novelties here are (1) a new observation that allows us to reduce the number of co-nondeterministic bits from $\log n$ (as in \cite{GrochowLevetWL}) down to $\log \log n$, and (2) using an \\ $\ACz(\DTISPpll)$ circuit for order finding, rather than $\FOLL$ as in \cite{ChattopadhyayToranWagner}.

\item \algprobm{Membership} for nilpotent groups is in $\ACz(\NTISPpll) \subseteq \FOLL \cap\qACz$ (Prop.~\ref{prop:Nilpotent}).

\end{itemize}

\subsection{Methods}
Several of our results involve careful analysis of the low-level circuit complexity of extant algorithms, showing that they in fact lie in smaller complexity classes than previously known. One important ingredient here is that we use simultaneous time- and space-restricted computations. This not only facilitates several proofs and gives better complexity bounds, but also gives rise to new algorithms such as for \algprobm{Membership} for nilpotent groups, which previously was not known to be in \FOLL.

One such instance is in our improved bound for order-finding and exponentiation in a semigroup (Lem.~\ref{lem:fastexp}). The previous proof \cite{BKLM} (still state of the art 23 years later) used a then-novel and clever ``double-barrelled'' recursive approach to compute these in \FOLL. In contrast, our proof uses standard repeated doubling, noting that it can be done in $\DTISPpll \subseteq \FOLL \cap \qACz$, recovering their result with standard tools and reducing the depth; in fact, from our proof we get $\qACz$ circuits of depth 2, which is clearly optimal from the perspective of depth.
 We use this improved bound on order-finding to improve the complexity of isomorphism testing of Abelian groups (Thm.~\ref{thm:abelian}), simple groups (Cor.~\ref{cor:simple}), and groups of almost all orders (Thm.~\ref{thm:ParallelDW}).

For a few of our results, however, we need to develop new tools to work with quasigroups. 
In particular,  for the $\qACz$ upper bound on \algprobm{MGS} for quasigroups, we cannot directly adapt the technique of Chattopadhyay, Torán, and Wagner. 
Indeed, their analysis of their algorithm already seems tight to us in terms of having depth $\Theta(\log \log n)$.

The first key is Thm.~\ref{thm:main-cqm}, putting \algprobm{Membership} for quasigroups into $\NTIME(\mathrm{polylog}(n))$ (more precisely, into $\betacc{2}\DTISPpll$). To do this, we replace the use of a cube generating sequences from \cite{ChattopadhyayToranWagner} with something that is nearly as good for the purposes of \algprobm{MGS}: we extend the Babai--Szemer\'edi  Reachability Lemma \cite[Thm.~3.1]{BabaiSzemeredi} from groups (its original setting) to quasigroups in order to obtain what we call \emph{cube-like} generating sequences, which also give us short straight-line programs. Division in quasigroups is somewhat nuanced, e.g., despite the fact that for any $a,b$, there exists a unique $x$ such that $ax = b$, this does not necessarily mean that there is an element ``$a^{-1}$'' such that $x=a^{-1}b$, because of the lack of associativity. Our proof is thus a careful adaptation of the technique of Babai \& Szemer\'edi, with a few quasigroup twists that result in a slightly worse, but still sufficient, bound.

\subsection{Prior work}
\paragraph{Isomorphism testing.} The best known runtime bound for \algprobm{GpI} is $n^{(1/4) \log_{p}(n) + O(1)}$-time (where $p$ is the smallest prime dividing $n$) due to Rosenbaum \cite{Rosenbaum2013BidirectionalCD} and Luks \cite{LuksCompositionSeriesIso} (see \cite[Sec. 2.2]{GR16}), though this tells us little about parallel complexity. In addition to Tarjan's result mentioned above \cite{MillerTarjan}, Lipton, Snyder, \& Zalcstein \cite{LiptonSnyderZalcstein} independently observed that if a group is $d$-generated, then we can decide isomorphism by considering all possible $d$-element subsets. This is the \textit{generator enumeration} procedure. Using the fact that  every group admits a generating set of size $\leq \log_{p}(n)$ (where $p$ is the smallest prime dividing $n$), Tarjan obtained a bound of $n^{\log_{p}(n) + O(1)}$-time for \algprobm{Group Isomorphism}, while Lipton, Snyder, \& Zalcstein \cite{LiptonSnyderZalcstein} obtained a stronger bound of $\mathsf{DSPACE}(\log^{2} n)$. Miller \cite{MillerTarjan} extended Tarjan's observation to the setting of quasigroups. There has been subsequent work on improving the parallel complexity of generator enumeration for quasigroups, resulting in bounds of $\betacc{2}\AC^{1}$ ($\AC^{1}$ circuits that additionally receive $O(\log^{2} n)$ non-deterministic bits, denoted by other authors as $\beta_{2}\AC^1$) due to Wolf
\cite{Wolf},\footnote{See footnote~\ref{fn:wolf}.} $\betacc{2}\SAC^{1}$ due to Wagner \cite{WagnerThesis}, and $\betacc{2}\LogSpace \cap \betacc{2}\FOLL$ due to Chattopadhyay, Tor\'an, \& Wagner \cite{ChattopadhyayToranWagner}. In the special case of groups, generator enumeration is also known to belong to $\betacc{2}\mathsf{SC}^{2}$ \cite{TangThesis}. There has been considerable work on polynomial-time, isomorphism tests for several families of groups, as well as more recent work on $\mathsf{NC}$ isomorphism tests---we refer to recent works \cite{GQcoho, DietrichWilson, GrochowLevetWL} for a survey. We are not aware of work on isomorphism testing for specific families of quasigroups that are not groups.

\paragraph{Min Generating Set.} As every (quasi)group has a generating set of size $\leq \lceil \log n \rceil$, \algprobm{MGS} admits an $n^{\log(n) + O(1)}$-time solution for (quasi)groups. In the case of groups, Arvind \& Tor\'an  \cite{ArvindToran} improved the complexity to $\mathsf{DSPACE}(\log^{2} n)$. They also gave a polynomial-time algorithm in the special case of nilpotent groups. Tang further improved the general bound for $\algprobm{MGS}$ for groups to $\betacc{2}\mathsf{SC}^{2}$ \cite{TangThesis}. We observe that Wolf's technique for placing \algprobm{Quasigroup Isomorphism} into $\mathsf{DSPACE}(\log^2 n)$ also suffices to get \algprobm{MGS for Quasigroups} into the same class. Very recently, Das \& Thakkar \cite{DasThakkar} improved the algorithmic upper bound for \algprobm{MGS} in the setting of groups to $n^{(1/4) \log(n) + O(1)}$. A month later, Lucchini \& Thakkar \cite{LucchiniThakkar} placed  \algprobm{MGS} for groups into $\mathsf{P}$. Prior to \cite{LucchiniThakkar}, \algprobm{MGS for Groups} was considered comparable to \algprobm{Group Isomorphism} in terms of difficulty. Our $\AC^{1}(\LogSpace)$ bound  (\Thm{thm:main-mgsq}) further closes the gap between \algprobm{Membership Testing} in groups and \algprobm{MGS for Groups}, and in particular suggests that \algprobm{MGS} is of comparable difficulty to \algprobm{Membership} for groups rather than \algprobm{GpI}. Note that \algprobm{Membership} for groups has long been known to belong to $\LogSpace$ \cite{BarringtonMcKenzie, Reingold}.

\section{Preliminaries}

\subsection{Algebra and Combinatorics}\label{sec:algebraprelims}

\paragraph{Graph Theory.} A \textit{strongly regular graph} with parameters $(n, k, \lambda, \mu)$ is a simple, undirected $k$-regular, $n$-vertex graph $G(V, E)$ where any two adjacent vertices share $\lambda$ neighbors, and any two non-adjacent vertices share $\mu$ neighbors. The complement of a strongly regular graph is also strongly regular, with parameters $(n, n-k-1, n-2-2k+\mu, n-2k+\lambda)$. \\

\noindent A \textit{magma} $M$ is an algebraic structure together with a binary operation $\cdot : M \times M \to M$. We will frequently consider subclasses of magmas, such as groups, quasigroups, and semigroups. 

\paragraph{Quasigroups and Latin squares.} A \textit{quasigroup} consists of a set $G$ and a binary operation $\star : G \times G \to G$ satisfying the following. For every $a, b \in G$, there exist unique $x, y$ such that $a \star x = b$ and $y \star a = b$.  We write $x = a\backslash b$ and $y = b/a$, i.\,e., $a\star (a\backslash b) = b$ and $ (b/a) \star a = b$. When the multiplication operation is understood, we simply write $ax$ for $a \star x$. A \emph{sub-quasigroup} of a quasigroup is a subset that itself is a quasigroup. This means it is closed under the multiplication as well as under left and right quotients. Given $X \subseteq G$, the sub-quasigroup generated by $X$ is denoted as $ \langle X \rangle$. It is the smallest sub-quasigroup containing $X$.

Unless otherwise stated, all quasigroups are assumed to be finite and represented using their Cayley (multiplication) tables.

As quasigroups need not be associative, the parenthesization of a given expression may impact the resulting value.  For a sequence $S := (s_{0}, s_{1}, \ldots, s_{k})$ and parenthesization $P$ from a quasigroup, define:
\[
\text{Cube}(S) = \{ P(s_{0}s_{1}^{e_{1}} \cdots s_{k}^{e_{k}}) : e_{1}, \ldots, e_{k} \in \{0,1\} \}.
\]

We say that $S$ is a \textit{cube generating sequence} if each element $g$ in the quasigroup can be written as $g = P(s_{0}s_{1}^{e_{1}} \cdots s_{k}^{e_{k}})$, for $e_{1}, \ldots, e_{k} \in \{0,1\}$. Here, $s_{i}^{0}$ indicates that $s_{i}$ is not being considered in the product.  For every parenthesization, every quasigroup is known to admit a cube generating sequence of size $O(\log n)$ \cite{ChattopadhyayToranWagner}.

A \textit{Latin square} of order $n$ is an $n \times n$ matrix $L$ where each cell $L_{ij} \in [n]$, and each element of $[n]$ appears exactly once in a given row or a given column. Latin squares are precisely the Cayley tables corresponding to quasigroups. We will abuse notation by referring to a quasigroup and its multiplication table interchangeably. An \textit{isotopy} of Latin squares $L_{1}$ and $L_{2}$ is an ordered triple $(\alpha, \beta, \gamma)$, where $\alpha, \beta, \gamma : L_{1} \to L_{2}$ are bijections satisfying the following: whenever $ab = c$ in $L_{1}$, we have that $\alpha(a)\beta(b) = \gamma(c)$ in $L_{2}$. Alternatively, we may view $\alpha$ as a permutation of the rows of $L_{1}$, $\beta$ as a permutation of the columns of $L_{1}$, and $\gamma$ as a permutation of the values in the table. Here, $L_{1}$ and $L_{2}$ are isotopic precisely if $x$ is the $(i,j)$ entry of $L_{1}$ if and only if $\gamma(x)$ is the $(\alpha(i), \beta(j))$ entry of $L_{2}$.

Albert showed that a quasigroup $Q$ is isotopic to a group $G$ if and only if $Q$ is isomorphic to $G$. In general, isotopic quasigroups need not be isomorphic \cite{Albert}.

A Latin square $L$ can equivalently be viewed as a set of triples $\{(i,j,L_{ij}) : i,j \in [n]\} \subseteq [n] \times [n] \times [n]$. Given a set of triples $S \subseteq [n] \times [n] \times [n]$, the Latin square property can equivalently be rephrased as: every $i \in [n]$ appears as the first---resp. second, resp. third---coordinate of some triple in $S$, and no two triples in $S$ agree in more than one coordinate. From this perspective, an additional potential symmetry of Latin squares emerges: two Latin squares $L_1, L_2$ are \emph{parastrophic}\footnote{This terminology is borrowed from the quasigroup literature. In the Latin square literature this is sometimes referred to as ``conjugate'', but we find the argument in Keedwell and Denes \cite[pp. 15--16]{KeedwellDenes} to use  the quasigroup nomenclature even in the setting of Latin squares compelling.} if there is a permutation $\pi \in S_3$ (where $S_3$ is the symmetric group of degree $3$) such that, when viewed as sets of triples, we have $L_2 = \{(x_{1^\pi}, x_{2^\pi}, x_{3^\pi}) : (x_1, x_2, x_3) \in L_1\}$; the induced map $L_1 \to L_2$ is called a \emph{parastrophy}. A \emph{main class isomorphism} of Latin squares is the composition of a parastrophy and an isotopy; if there exists a main class isomorphism $L_1 \to L_2$, we say they are \emph{main class isomorphic}.\footnote{Miller \cite{MillerTarjan}, and then Levet \cite{LevetLatinSquares} following Miller, referred to this as ``main class isotopy''; we have since found several modern textbooks on quasigroups and Latin squares that refer to this notion as one of ``main class isomorphism'', ``paratopy'', or ``isostrophy'' (sic). Keedwell and Denes \cite[pp. 15--16]{KeedwellDenes} have a nice discussion of the  terminology, as well as the history of its usage.}

For a given Latin square $L$ of order $n$, we associate a \textit{Latin square graph} $G(L)$ that has $n^{2}$ vertices; one for each triple $(a, b, c)$ that satisfies $ab = c$. Two vertices $(a, b, c)$ and $(x, y, z)$ are adjacent in $G(L)$ precisely if $a = x$ or $b = y$ or $c = z$. Miller showed that two Latin square graphs $G_{1}$ and $G_{2}$ are isomorphic if and only if the corresponding Latin squares, $L_{1}$ and $L_{2}$, are main class isomorphic \cite{MillerTarjan}. 

A Latin square graph on $n^{2}$ vertices is a strongly regular graph with parameters $(n^{2}, 3(n-1), n, 6)$. Conversely, a strongly regular graph with these same parameters $(n^{2}, 3(n-1), n, 6)$ is called a \textit{pseudo-Latin square graph}. Bruck showed that for $n > 23$, a pseudo-Latin square graph is a Latin square graph \cite{Bruck}.

\paragraph{Group Theory.} For a standard reference, see \cite{Robinson1982}. All groups are assumed to be finite. For a group $G$, $d(G)$ denotes the minimum size of a generating set for $G$. The \textit{Frattini subgroup} $\Phi(G)$ is the set of non-generators of $G$. If $G$ is a $p$-group, the Burnside Basis Theorem (see \cite[Theorem 5.3.2]{Robinson1982}) provides that (i) $G = G^{p}[G,G]$, (ii) $G/\Phi(G) \cong (\mathbb{Z}/p\mathbb{Z})^{d(G)}$, and (iii) if $S$ generates $(\mathbb{Z}/p\mathbb{Z})^{d(G)}$, then any lift of $S$ generates $G$. A \textit{chief series} of $G$ is an ascending chain $(N_{i})_{i=0}^{k}$ of normal subgroups of $G$, where $N_{0} = 1$, $N_{k} = G$, and each $N_{i+1}/N_{i}$ ($i = 0, \ldots, k-1$) is minimal normal in $G/N_{i}$. 
For $g, h \in G$, the \textit{commutator} $[g,h] := ghg^{-1}h^{-1}$. The \textit{commutator subgroup} $[G,G] = \langle \{ [g,h] : g,h \in G \} \rangle$.

\paragraph{Designs.} Let $t \leq k \leq v$ and $\lambda$ be positive integers. A $(t, k, \lambda, v)$ design is an incidence structure $\mathcal{D} = (X, \mathcal{B}, I)$, where $X$ is a set of $v$ points, $\mathcal{B}$ is a subset of $\binom{X}{k}$---whose elements are referred to as \textit{blocks}---and such that each $t$-element subset of $X$ belongs to exactly $\lambda$ blocks. Now $I$ is the point-block incidence matrix, where $I_{x,B} = 1$ precisely if the point $x$ belongs to the block $B$.

If $t < k < v$, we say that the design is \textit{non-trivial}. If $\lambda = 1$, the design is referred to as a \textit{Steiner design}. We denote Steiner designs as $(t, k, v)$-designs when we want to specify $v$ the number of points, or Steiner $(t, k)$-designs when referring to a family of designs. We note that Steiner $(2, 3)$-designs are known as \textit{Steiner triple systems}. Projective planes are Steiner $(2, q+1, q^{2} + q + 1)$-designs, and affine planes are Steiner $(2, q, q^{2})$-designs. We assume that designs are given by the point-block incidence matrix.

For a design $\mathcal{D} = (X, \mathcal{B}, I)$, we may define a \textit{block intersection graph} (also known as a \textit{line graph}) $G(V, E)$, where $V(G) = \mathcal{B}$ and two blocks $B_{1}, B_{2}$ are adjacent in $G$ if and only if $B_{1} \cap B_{2} \neq \emptyset$. In the case of a Steiner $2$-design, the block-intersection graph is strongly regular. For Steiner triple systems, the block-intersection graphs are strongly regular with parameters $(n(n-1)/6, 3(n-3)/2, (n+3)/2, 9)$. Conversely, strongly regular graphs with the parameters $(n(n-1)/6, 3(n-3)/2, (n+3)/2, 9)$ are referred to as \textit{pseudo-STS graphs}. Bose showed that pseudo-STS graphs with strictly more than $67$ vertices are in fact STS graphs \cite{Bose}.

\paragraph{Algorithmic Problems.} 
A multiplication table of a magma $G = \{g_1, \dots, g_n\}$ of order $n$ is an array $M$ of length $n^2$ where each entry $M[i+(j-1)n]$ for $i,j \in \{1, \dots, n\}$ contains the binary representation of $k$ such that $g_ig_j= g_k$. In the following all magmas are given as their multiplication tables.

We will consider the following algorithmic problems. The \algprobm{Quasigroup Isomorphism} problem takes as input two quasigroups $Q_{1}, Q_{2}$ given by their multiplication tables, and asks if there is an isomorphism $\varphi : Q_{1} \cong Q_{2}$. The \algprobm{Membership} problem for groups takes as input a group $G$ given by its multiplication table, a set $S \subseteq G$, and an element $x \in G$, and asks if $x \in \langle S \rangle$. We may define the \algprobm{Membership} problem analogously when the input is a semigroup or quasigroup, and $\langle S \rangle$ is considered as the sub-semigroup or sub-quasigroup, respectively.

The \algprobm{Minimum Generating Set} (\algprobm{MGS}) problem takes as input a magma $M$ given by its multiplication table and asks for a generating set $S \subseteq M$ where $|S|$ is minimum. At some point we will also consider the \emph{decision variant} of \algprobm{MGS}: here we additionally give an integer $k$ in the input and the question is whether there exists a generating set of cardinality at most $k$.

We will be primarily interested in \algprobm{Minimum Generating Set} and \algprobm{Membership} in the setting of (quasi)groups,  with a few excursions to semigroups and magmas. Note that \algprobm{Membership} for groups is known to be in $\LogSpace \cap \qACz$ \cite{BarringtonMcKenzie, Reingold, Fleischer}. Here the containment in \LogSpace follows from the deep result by Reingold \cite{Reingold} that symmetric logspace (nondeterministic logspace where each transition is also allowed to be applied backward) coincides with \LogSpace.

\subsection{Computational Complexity} \label{sec:Complexity}
We assume that the reader is familiar with standard complexity classes such as \LOGSPACE, \NL, \NP, and \EXP. For a standard reference on circuit complexity, see \cite{VollmerText}. We consider Boolean circuits using the $\mathsf{AND}, \mathsf{OR}, \mathsf{NOT},$ and $\mathsf{Majority},$ where $\mathsf{Majority}(x_{1}, \ldots, x_{n}) = 1$ if and only if $\geq n/2$ of the inputs are $1$. Otherwise, $\mathsf{Majority}(x_{1}, \ldots, x_{n}) = 0$. In this paper, we will consider $\mathsf{DLOGTIME}$-uniform circuit families $(C_{n})_{n \in \mathbb{N}}$. For this,
one encodes the gates of each circuit $C_n$ by bit strings of length $O(\log n)$. Then the circuit family $(C_n)_{n \geq 0}$
is called \emph{\DLOGTIME-uniform}  if (i) there exists a deterministic Turing machine that computes for a given gate $u \in \{0,1\}^*$
of $C_n$ ($|u| \in O(\log n)$) in time $O(\log n)$ the type of gate $u$, where the types are $x_1, \ldots, x_n$, $\mathsf{NOT}, \mathsf{AND}, \mathsf{OR}, \mathsf{Majority},$ or oracle gates,
and (ii) there exists a deterministic Turing machine that decides for two given gates $u,v \in \{0,1\}^*$
of $C_n$ ($|u|, |v| \in O(\log n)$) and a binary encoded integer $i$ with $O(\log n)$ many bits
in time $O(\log n)$ whether $u$ is the $i$-th input gate for $v$.

\begin{definition}
Fix $k \geq 0$. We say that a language $L$ belongs to (uniform) $\mathsf{NC}^{k}$ if there exist a (uniform) family of circuits $(C_{n})_{n \in \mathbb{N}}$ over the $\mathsf{AND}, \mathsf{OR}, \mathsf{NOT}$ gates such that the following hold:
\begin{itemize}
\item The $\mathsf{AND}$ and $\mathsf{OR}$ gates take exactly $2$ inputs. That is, they have fan-in $2$.
\item $C_{n}$ has depth $O(\log^{k} n)$ and uses (has size) $n^{O(1)}$ gates. Here, the implicit constants in the circuit depth and size depend only on $L$.

\item $x \in L$ if and only if $C_{|x|}(x) = 1$. 
\end{itemize}
\end{definition}

\noindent The complexity class $\AC^{k}$ is defined analogously as $\mathsf{NC}^{k}$, except that the $\mathsf{AND}, \mathsf{OR}$ gates are permitted to have unbounded fan-in.
That is, a single $\mathsf{AND}$ gate can compute an arbitrary conjunction, and a single $\mathsf{OR}$ gate can compute an arbitrary disjunction. The class $\SAC^k$ is defined analogously, in which the $\mathsf{OR}$ gates have unbounded fan-in but the $\mathsf{AND}$ gates must have fan-in $2$.
The complexity class $\mathsf{TC}^{k}$ is defined analogously as $\AC^{k}$, except that our circuits are now permitted $\mathsf{Majority}$ gates of unbounded fan-in.
We also allow circuits to compute functions by using multiple output gates. 

Furthermore, for a language $L$ the class $\AC^k(L)$, apart from Boolean gates, also allows oracle gates for $L$ (an oracle gate outputs $1$ if and only if its input is in $L$).
If $K\in \AC^k(L)$, then $K$ is said to be $\AC^k$-Turing reducible to $L$.
Finally, for some complexity class $\mathcal{C}$ denote $\AC^k(\mathcal{C})$ to be the set of decision problems that are $\AC^k$-Turing reducible to problems in $\mathcal{C}$. 
Be aware that here we follow the notation of \cite{VollmerText}, which is different from \cite{WagnerThesis,GrochowLevetWL} (where $\AC^k(\mathcal{C})$ is used to denote composition of functions).

For every $k$, the following containments are well-known:
\[
\mathsf{NC}^{k} \subseteq \SAC^k \subseteq  \AC^{k} \subseteq \mathsf{TC}^{k} \subseteq \mathsf{NC}^{k+1}.
\]

\noindent In the case of $k = 0$, we have that:
\[
\mathsf{NC}^{0} \subsetneq \AC^{0} \subsetneq \mathsf{TC}^{0} \subseteq \mathsf{NC}^{1} \subseteq \LogSpace \subseteq \mathsf{NL} \subseteq \mathsf{SAC}^{1} \subseteq \AC^{1}.
\]

\noindent We note that functions that are $\mathsf{NC}^{0}$-computable can only depend on a bounded number of input bits. Thus, $\mathsf{NC}^{0}$ is unable to compute the $\mathsf{AND}$ function. It is a classical result that $\AC^{0}$ is unable to compute \algprobm{Parity} \cite{Smolensky87algebraicmethods}. The containment $\mathsf{TC}^{0} \subseteq \mathsf{NC}^{1}$ (and hence, $\mathsf{TC}^{k} \subseteq \mathsf{NC}^{k+1}$) follows from the fact that $\mathsf{NC}^{1}$ can simulate the $\mathsf{Majority}$ gate. 

We will crucially use the following throughout the paper.

\begin{theorem}[{\cite[Theorem 5.1]{HESSE2002695}}] \label{lem:multiplication}
The product of $(\log n)^{O(1)}$-many integers each of $(\log n)^{O(1)}$ bits can be computed in $\ACz$.
\end{theorem}

We also use the following easy lemma:

\begin{lemma} \label{lem:factor}
The prime factors of a $O(\log n)$-bit integer can be computed in \ACz.
\end{lemma}

(We remind the reader that we always refer to the uniform classes unless otherwise specified; without uniformity the result would be immediate as all functions of $\log n$ bits are in nonuniform $\ACz$.)

\begin{proof}
Each prime divisor of $n$ can be represented using $O(\log n)$ bits. As the numbers involved are only $O(\log n)$ bits, ordinary arithmetic function of these numbers can be computed in $\mathsf{AC}^0$, in particular, testing if an $O(\log n)$-bit number is prime, testing if one $O(\log n)$-bit number divides another. So, in parallel, for all numbers $x=2,3,\dotsc,n/2$, an $\mathsf{AC}^0$ circuit checks which ones are prime and divide $n$.
\end{proof}

\paragraph{Further circuit classes.} The complexity class $\MACz$ is the set of languages decidable by constant-depth uniform circuit families with a polynomial number of $\mathsf{AND}, \mathsf{OR},$ and $\mathsf{NOT}$ gates, and a single $\mathsf{Majority}$ gate at the output. The class $\MACz$ was introduced (but not so named) in \cite{VotingPolynomials}, where it was shown that $\MACz \subsetneq \mathsf{TC}^{0}$. This class was subsequently given the name $\MACz$ in \cite{LearnabilityAC0}.

The complexity class $\mathsf{FOLL}$ is the set of languages decidable by uniform $\AC$ circuit families of depth $O(\log \log n)$ and polynomial size. It is known that $\AC^{0} \subsetneq \mathsf{FOLL} \subsetneq \AC^{1}$, and it is open as to whether $\mathsf{FOLL}$ is contained in $\mathsf{NL}$ \cite{BKLM}.

We will be particularly interested in $\mathsf{NC}$ circuits of quasipolynomial size (i.\,e., $2^{O(\log^k n)}$ for some constant $k$). For a circuit class $\mathcal{C} \subseteq \mathsf{NC}$, the analogous class permitting a quasipolynomial number of gates is denoted $\mathsf{quasi}\mathcal{C}$. We will focus specifically on \qACz. Here, we stress that our results for \qACz will be stated for the nonuniform setting. Note that \DLOGTIME uniformity does not make sense for \qACz, as we cannot encode gate indices using $O(\log n)$ bits. Nonetheless, there exist suitable notions of uniformity for $\qACz$ \cite{BarringtonQuasipolynomial,FerrarottiGonzalezScheweTurull}.

\paragraph{Bounded nondeterminism.} For a complexity class $\mathcal{C}$, we define $\betacc{i}\mathcal{C}$ to be the set of languages $L$ such that there exists an $L' \in \mathcal{C}$ such that $x \in L$ if and only if there exists $y$ of length at most $O(\log^{i} |x|)$ such that $(x, y) \in L'$. Similarly, define $\alphacc{i}\mathcal{C}$ to be the set of languages $L$ such that there exists an $L' \in \mathcal{C}$ such that $x \in L$ if and only if for all $y$ of length at most $O(\log^{i} |x|)$, $(x,y) \in L'$. For any $i \geq 0$ and any $c \geq 0$, both $\betacc{i}\FOLL$ and $\alphacc{i}\FOLL$ are contained in  $\mathsf{quasi}\FOLL$, and so cannot compute \algprobm{Parity} \cite{ChattopadhyayToranWagner, Smolensky87algebraicmethods}. Note that $\alphacc{1}\mathcal{C} \cup  \betacc{1}\mathcal{C}\subseteq \ACz(\mathcal{C})$.

\paragraph{Time and space-restricted Turing machines.}
When considering complexity classes defined by Turing machines with a time bound $t(n) \in o(n)$, we use Turing machines with random access and a separate address (or \emph{index}) tape. After writing an address, the machine can go to a query state reading the symbol from the input at the location specified by the address tape. As usual, the machines are otherwise allowed to have multiple work tapes. 

For functions $t(n), s(n) \in \Omega(\log n)$, the classes $\DTISP(t(n),s(n)) $ and $\NTISP(t(n),s(n)) $ are defined to consist of decision problems computable by deterministic (resp.\ nondeterministic) $t(n)$ time and $s(n)$ space bounded Turing machines. Be aware that there must be one Turing machine that simultaneously satisfies the time and space bound. For details we refer to \cite[Section 2.6]{VollmerText}. For further reading on the connection to \qACz, we refer to \cite{BarringtonQuasipolynomial,FerrarottiGonzalezScheweTurull}.

We frequently use $\DTISPpll$ and $\NTISPpll$. However, because of how small the time and space bounds are for these classes, when we abuse notation to say some function (not necessarily decision problem) is computable in $\DTISPpll$, there is some ambiguity as to what this might mean. We use the following two equivalent definitions. First we need some setup. For a function $f$ from bit-strings to bit-strings, we define its \emph{bit function} $\text{bit-}f$ as follows. Write $f(x)_i$ to denote the $i$-th bit of $f(x)$. Then define
\[
\text{bit-}f(x,i) = \begin{cases}
f(x)_i & i \leq |f(x)| \\
\bot & i > |f(x)|.
\end{cases}
\]
Note that $\text{bit-}f$ is always a function which we may take to have at most two output bits (e.g., by encoding $f(x)_i$ by repeating the bit, and encoding $\bot$ by $01$), i.e. a pair of decision problems.

\begin{deflem} \label{def:DTISP}
For any function $f \colon \{0,1\}^* \to \{0,1\}^*$ with $|f(x)| \leq O(\log |x|)$, the following are equivalent:
\begin{enumerate}
\item $f$ is computable by a Turing machine, using $\poly(\log n)$ time and $O(\log n)$ space, that halts with the result written in a specified place on its work tape.

\item The two bits of $\text{bit-}f$ are each decision problems in $\DTISPpll$, i.\,e., can each be decided by a Turing machine using $\poly(\log n)$ time and $O(\log n)$ space. 
\end{enumerate}
In either case, we say $f$ is computable in $\DTISPpll$.
\end{deflem}

\begin{proof}
($\Rightarrow$) Suppose $f$ is computable by a Turing machine $M$ using $\poly(\log n)$ time and $O(\log n)$ workspace (\emph{including} the output). Then the following machine $N$ computes the function $(x,i) \mapsto f(x)_i$. On input $(x,i)$, $N$ simulates $M(x)$ while leaving $i$ untouched. After the simulation has completed, one of $N$'s worktapes contains $f(x)_i$, while $i$ is still on the input tape. Next, using another worktape to keep a counter, $N$ walks through $f(x)$ until it gets to the $i$-th bit, and outputs that bit of $f(x)$. 

The entirety of the computation after the simulation of $N$ uses an amount of space that is $|i| \leq O(|f(x)|) \leq O(\log |x|)$, and time that is at most $|i|^2 \leq O(\log^2 |x|)$ (in fact, by noting that incrementing a counter by $+1$ $k$ many times only incurs $O(k)$ many changes to the bits of the counter, this can be reduced to $O(\log |x|)$ as well).

($\Leftarrow$) Suppose the function $(x,i) \mapsto f(x)_i$ is computable by a Turing machine $N$ using $\poly(\log n) = \poly(\log |x|)$ time and $O(\log n) = O(\log |x|)$ space. The following machine $M$ computes $f$ in the manner in which, when started with $x$ on its input tape, at the end of the computation $f(x)$ is the only thing written on the work tape. On input $x$, $M$ simulates $N(x,0), N(x,1), \dotsc$, writing those bits in the specified place on its work tape, until it reaches an $i$ such that $N(x,i) = \bot$. As $N$ was called $|f(x)|+1$ times, this multiplies the time complexity of $N$ by $|f(x)| \leq O(\log |x|)$. However, the space need does not increase except to keep track of the integer $i$ for which the machine will next simulate $N(x,i)$, plus the space for the output, as the space to simulate $N(x,i)$ can be re-used when simulating $N(x,i+1)$. Thus the space used is that of $N$, plus the space for the output, plus the space to maintain $i$, which is $|i| \leq \log (|f(x)|+1) \leq O(\log \log |x|)$, so the total space is still $O(\log |x| + \log \log|x|) = O(\log |x|)$. The additional clearing and marking at the end at most doubles the time and adds a constant amount of space usage.
\end{proof}

\begin{fact}\label{fact:NTISP} 
$\NTISPpll \subseteq \NTIME(\mathrm{polylog} (n)) \subseteq \qACz$. In particular, any decision problem in $\NTIME(\mathrm{polylog} (n))$ is computable by 
a quasi-polynomial-size DNF (a particular kind of depth-2 $\qACz$ circuit).

Furthermore, decision problems in $\DTIME(\mathrm{polylog} (n))$ are computable by quasi-polynomial-size CNFs as well.
\end{fact}

\begin{proof}
Given a $\mathrm{polylog} (n)$-time nondeterministic Turing machine $M$, we get a decision tree $T_M$ that decides that same language as follows. Each node of the tree corresponds to any time $M$ either queries an input bit or makes a nondeterministic guess. The two children of such a node correspond to whether the queried (resp., nondeterministically guessed) bit was 0 or 1. The leaves of the tree are labeled accepting or rejecting according to whether $M$ accepts or rejects at that point (note: once no more input bits are queried and no more nondeterministic guesses are made, the answer is fully determined, even though $M$ may deterministically spend an additional $\mathrm{polylog} (n)$ steps figuring out what the answer should be). The decision tree has height at most $\mathrm{polylog} (n)$, and thus at most quasi-polynomially many branches. We get a DNF by taking the disjunction over all accepting paths of the conjunction of the answers to the input queries on those paths.

If there are no non-deterministic guesses made, i.e.\ our machine $M$ was deterministic, then we can get a CNF from the above decision tree by using De Morgan's law: negate the leaves of the decision tree, get the corresponding DNF of the negated tree, then negate the DNF. (This does not work in the nondeterminstic case, as the disjunction includes all inputs such that there \emph{exists} nondeterministic guesses that lead to an accepting leaf, and when we negate the leaves we end up with co-nondeterminism instead of non-determinism.)
\end{proof}

\begin{lemma}\label{lem:DTISPFoll}
   $ \NTISPpll \subseteq \FOLL$.
\end{lemma}
\begin{proof}
    This is essentially the proof of Savitch's Theorem: A configuration of a Turing machine consist of the current state, the work and index tape content, but \emph{not} the content of the input tape. For configurations $\alpha, \beta $ define the $\Reach$ predicate as follows:
 	\begin{align*}
 \Reach(\alpha, \beta, 0)
 		&\iff \beta \text{ is reachable from $\alpha$ in at most one computation step} \\
 \Reach(\alpha, \beta, \mathrlap{i}\phantom{0})
 		&\iff \exists \gamma: \big( \Reach(\alpha, \gamma, i-1) \wedge
 		\Reach(\gamma, \beta, i-1) \big)		
 	\end{align*}
	for $i \geq 1$.
 Thus, $\Reach(\alpha, \beta, i)$ holds if and only if $\beta$ can be reached from $\alpha$ in at most $2^i$ computation steps.
 	Since the running time is bounded by $\log^k n$ for some $k$, by letting $\Start(w)$  denote the initial configuration for the input $w \in \Sigma^*$ and $\Accept$ the accepting configuration of $M$ (which can be assumed to be unique after a suitable manipulation of $M$), we have 
 	\[ \Reach(\Start(w), \Accept, \log \log^k n) \Longleftrightarrow w \in L.\]
  Now, it remains to observe that the inductive definition of the $\Reach$ predicate can be evaluated in \FOLL since the recursion depth is $\Oh(\log\log n)$, each recursion step is clearly in $\ACz$ and there are only polynomially many configurations (because of the space bound of $\log n$).
\end{proof}

\noindent By the very definition we have $\DTISPpll \subseteq \LOGSPACE$ and $\NTISPpll \subseteq \NL$. Thus, we obtain
\begin{fact}\
   \begin{itemize}
    \item $\ACz(\DTISPpll) \subseteq \LOGSPACE \cap \FOLL \cap \qACz$ and
    \item $\ACz(\NTISPpll) \subseteq \NL \cap \FOLL \cap \qACz$.
\end{itemize} 
\end{fact}

\paragraph{Disjunctive truth-table reductions.} We finally recall the notion of a disjunctive truth-table reduction. Again let $L_{1}, L_{2}$ be languages. We say that $L_{1}$ is \textit{disjunctive truth-table (dtt)} reducible to $L_{2}$, denoted $L_{1} \leq_{dtt} L_{2}$, if there exists a function $g$ mapping a string $x$ to a tuple of strings $(y_{1}, \ldots, y_{k})$ such that $x \in L_1$ if and only if there is some $i \in \{1,\dots, k\}$ such that $y_i \in L_2$. When $g$ is computable in a complexity class $\mathcal{C}$, we call this a $\mathcal{C}$-dtt reduction, and write $L_1 \leq_{dtt}^\mathcal{C} L_2$. For any class $\mathcal{C}$, $\mathcal{C}$-dtt reductions are intermediate in strength between $\mathcal{C}$-many-one reductions and $\mathcal{C}$-Turing reductions. 

\begin{fact} \label{fact:dtt}
$\betacc{2} \ACz(\mathcal{C})$ is closed under $\leq_{dtt}^\ACz$ reductions for any class $\mathcal{C}$.
\end{fact}

\begin{proof}
Let $L \in \betacc{2}\ACz(\mathcal{C})$, and suppose $L' \leq_{dtt}^\ACz L$ via $g$. We show that $L'$ is also in $\betacc{2}\ACz(\mathcal{C})$. Let $V \in \ACz(\mathcal{C})$ be a predicate such that $x \in L \iff [\exists^{\log^2 |x|} w] \,V(x,w)$ for all strings $x$.
Then we have $x \in L'$ iff $\exists\, i \in \{1, \dots, k\}, \exists\, w \in \{0,1\}^{\log^2 n}: V(g(x)_i, w)$, where $g(x)_i$ denotes the $i$-th string output by $g(x)=(y_1,\dotsc,y_k)$. Since $g$ is, in particular, polynomial size, we have $k \leq \poly(|x|)$, so we have the bit-length of the index $i$ is at most $O(\log|x|)$. Finally, since $\ACz(\mathcal{C})$ denotes the oracle class, it is closed under $\ACz$ reductions, and thus the function $(x,w,i) \mapsto V(g(x)_i, w)$ is in $\ACz(\mathcal{C})$, and therefore $L' \in \betacc{2}\ACz(\mathcal{C})$.
\end{proof}

\section{Order Finding and Applications}

In this section, we consider the parallel complexity of order finding. We begin with the following lemma.

\begin{lemma}\label{lem:fastexp}
    The following function is in $\DTISPpll$ (NB: Definition/Lemma~\ref{def:DTISP}):
    On input of a multiplication table of a semigroup $S$, an element $s \in S$, and a unary or binary number $k \in \N$ with $k \leq \abs{S}$, compute $s^k$. 
\end{lemma}

Note that, because the multiplication table is part of the input, the input size $n$ is always at least $|S|$.

\begin{proof}
    If $k$ is given in unary, we first compute its binary representation using a binary search (note that we can write it on the work tape as is uses at most $\ceil{\log\abs{S}}$ bits). We identify the semigroup elements with the natural numbers $0, \dots, \abs{S}-1$.
    Now, we compute $s^k$ using the standard fast exponentiation algorithm. Note that the multiplication of two semigroup elements can be done in $\DTIME(\log n)$ as we only need to write down two $\log n$ bit numbers on the address tape (if the multiplication table is not padded up to a power of two, this is still in $\DTISPpll$ because we need to multiply two $\log n$ bit numbers to compute the index in the multiplication table). 

    It is well-known that the fast exponentiation algorithm needs only $\Oh(\log k)$ elementary multiplications and $O(\log k + \log n)$ space; hence, the lemma follows.
\end{proof}

Note that \Lem{lem:fastexp} together with \Lem{lem:DTISPFoll} gives a new proof that the problem of computing a power $s^k$ in a semigroup can be done in \FOLL. This approach seems easier and more general than the double-barrelled recursive approach in \cite{BKLM}.

\Lem{lem:fastexp} also yields the following immediate corollary:

\begin{corollary} \label{cor:fastexp}\label{lem:order_finding}
On input of a multiplication table of a group $G$, an element $g \in G$, and $k \in \N$, we may decide whether  $\ord{g} = k$ in $\alphacc{1}\DTISPpll$. The same applies if, instead of $k$, another group element $h\in G$ is given with the question whether $\ord{g} = \ord{h}$. In particular, both questions can be decided by quasi-polynomial-size CNFs (a particular case of depth-2 $\qACz$).
\end{corollary}

\begin{proof}
First check whether $g^k = 1$ using \Lem{lem:fastexp}. If yes (and $k > 1$), we use $O(\log n)$ universally quantified co-nondeterministic bits to verify that for all $1 \leq i < k$ that $g^{i} \neq 1$ using again  \Lem{lem:fastexp}. For the second form of input observe that $\ord{g} = \ord{h}$ if and only if for all $i \leq \abs{G}$ we have $g^i = 1$ if and only if $h^i=1$. 

For the depth-2 upper bound: we can compute the decision problems $(g,i) \mapsto [g^i = 1]$, $(g,i) \mapsto [g^i \neq 1]$, and $(g,h,i) \mapsto [g^i=1 \leftrightarrow h^i=1]$ in $\DTISPpll$, which in turn can be decided by quasi-polynomial-size CNFs. Implementing the above strategy, the $\alphacc{1}$ becomes an AND gate of fan-in $\poly(n)$. That can be merged with the AND gates at the top of the aforementioned CNFs to get a single CNF of quasi-polynomial size. \qedhere
\end{proof}

\subsection{Application to isomorphism testing}
Using \Cor{cor:fastexp}, we can improve the upper bound for isomorphism testing of finite simple groups. Previously, this problem was known to be in $\LogSpace$ \cite{TangThesis} and $\mathsf{FOLL}$ \cite{GrochowLevetWL}. We obtain the following improved bound.

\begin{corollary} \label{cor:simple}
Let $G$ be a finite simple group and $H$ be arbitrary. We can decide isomorphism between $G$ and $H$ in $\ACz(\DTISPpll)$ 
\end{corollary}

\begin{proof}
 As $G$ is a finite simple group, $G$ is determined up to isomorphism by (i) $|G|$, and (ii) the set of orders of elements of $G$: $\text{spec}(G) := \{ \ord{g} : g \in G \}$ \cite{SimpleOrder}. We may check whether $|G| = |H|$ in $\AC^{0}$. By \Cor{cor:fastexp}, we may compute and compare $\text{spec}(G)$ and $\text{spec}(H)$ in $\ACz(\DTISPpll)$. \qedhere 
\end{proof}

In light of \Cor{cor:fastexp}, we obtain an improved bound for testing whether a group $G$ is nilpotent. Testing for nilpotency was previously known to be in  $\LogSpace \cap \mathsf{FOLL}$ \cite{BKLM}.

\begin{corollary}
Let $G$ be a finite group given by its multiplication table. We may decide whether $G$ is nilpotent in $\ACz(\DTISPpll)$.        
\end{corollary}

\begin{proof}
A finite group $G$ is nilpotent if and only if it is the direct product of its Sylow subgroups. As $G$ is given by its multiplication table, we may in $\ACz$ compute the prime factors of $|G|$ (Lemma~\ref{lem:factor}).  Thus, for each prime $p$ dividing $|G|$, we identify the elements $X_{p} := \{ g : \ord{g} \text{ is a power of } p \}$ and then test whether $X_{p}$ forms a group. By \Cor{cor:fastexp}, we can identify $X_{p}$ in $\ACz(\DTISPpll)$. Verifying the group axioms is $\ACz$-computable. The result follows.
\end{proof}

\subsection{Application to membership testing}
A group $G$ has the \emph{$\log n$ power basis property} (as defined in \cite{BKLM,Fleischer}) if for every subset $X \subseteq G$ every $g \in \langle X \rangle$ can be written as $g = g_1^{e_1} \cdots g_m^{e_m}$ with $m \leq \log n$ and suitable $g_i \in X$ and $e_i \in \Z$.

\begin{observation} \label{obs:logpowerbasis}
    Membership testing for semigroups with the $\log n$ power basis property is in $\NTISPpll$. 
\end{observation}

\begin{proof}
This follows by guessing suitable exponents and using \Lem{lem:fastexp} to compute the respective powers. In order to keep the space logarithmically bounded, we do not simply guess $g_1, \dotsc, g_m \in X$---which could take up to $\log^2 n$ bits to store---but rather guess the $g_i$ sequentially and only store the running product. That is, first guess $g_1$ and $e_1$ using $2 \log n$ bits, and compute $g_1^{e_1}$. Inductively suppose we have computed $g_1^{e_1} \dotsb g_i^{e_i}$. Note that we do not need to store the elements $g_1, \dotsc, g_i$ and exponents $e_1,\dotsc,e_i$, we only keep the product $g_1^{e_1} \dotsb g_i^{e_i}$, which takes $\log n$ bits. Then we guess $g_{i+1}$ and $e_{i+1}$, and continue.
\end{proof}

This observation allows us to improve several results from \cite{BKLM,Fleischer} from \FOLL to \\ \NTISPpll.

\begin{corollary} \label{cor:cayley_abelian}
\algprobm{Membership} for commutative semigroups is in $\NTISPpll$. 
\end{corollary}

\begin{proof} 
Commutative semigroups have the $\log n$ power basis property \cite[Lem.~4.2]{Fleischer}.
\end{proof}

\begin{corollary} 
Let $d$ be a constant. 
\algprobm{Membership} for solvable groups of class bounded by $d$ is in $\NTISPpll$. 
\end{corollary}

\begin{proof}
    Let $X$ be the input set for \algprobm{Membership} and $G$ the input group. 
Theorem 3.5 in \cite{BKLM} is proved by showing that $\langle X  \rangle = X_{d}$ where $X_0 = X$ and for some suitable constant $C$
\[X_i = \{x_1^{e_1} \cdots x_m^{e_m} \mid m \leq C \log n, x_j \in X_{i-1}, e_j \in \Z \text{ for } j \in \{1, \dots, m\}\}.\]

Now, as in \cref{obs:logpowerbasis}, we can guess an element of $X_i$ in \NTISPpll{} given that we have the elements of $X_{i-1}$. As $d$ is a constant, we can guess the elements of $X_{i-1}$ in \NTISPpll{} by induction. Thus, we can decide membership in   $\langle X  \rangle = X_{d}$   in \NTISPpll.
\end{proof}

For nilpotent groups, we can do even better and show \NTISPpll{} even without a bound on the nilpotency class---thus, improving considerably over \cite{BKLM}. This also underlines that the class \NTISPpll{} is useful not only because it provides a better complexity bound than \FOLL, but it also facilitates some proofs.
 
\begin{theorem} 
\algprobm{Membership} for nilpotent groups is in \NTISPpll.
\end{theorem}

\begin{proof}
 Let $X$ be the input set for \algprobm{Membership} and $G$ the input group. Write $n = \abs{G}$. Note that the nilpotency class of $G$ is at most $\log n$. Define $C = \{[x_1, \dots, x_\ell]\mid x_i \in X, \ell \leq \log n\}$ where $[x_1, \dots, x_\ell]$ is defined inductively by $[x_1, \dots, x_\ell] = [[x_1, \dots, x_{\ell-1}],x_\ell]$ for $\ell \geq 3$. When $\ell = 1$, define $[x_{1}, \ldots, x_{\ell}] := x_{1}$; and when $\ell = 2$, $[x_{1}, \ldots, x_{\ell}] := [x_{1}, x_{2}]$.

We claim that there is a subset $C'= \{c_1, \dots, c_m\} \subseteq C$ with $m \leq \log n$ such that every $g \in \langle X \rangle$ can be written as $g = c_1^{e_1} \cdots c_m^{e_m}$ with $e_i \in \Z$.

Let $\Gamma_m$ denote the $m$-th term of the lower central series of $\langle X \rangle$ (meaning that $\Gamma_0 = \langle X \rangle$ and $\Gamma_{m+1} = [\Gamma_m, \langle X \rangle]$). 
Then $\Gamma_m$ is generated by $C_m = \{[x_1, \dots, x_k]\mid x_i \in X, k\geq m\}$ (e.\,g., \cite[Lemma 2.6]{ClementMZ17}). 
Observe that, although $k$ is unbounded, the terms with $k >\log n$ are trivial because $\log n$ is a bound on the nilpotency class.
We obtain the desired set $C'$ by first choosing a minimal generating set for the Abelian group $\Gamma_0/\Gamma_1$, then a minimal generating set of $\Gamma_1/\Gamma_2$ and so on (meaning that $(c_1, \dots, c_m)$ is a so-called \emph{polycyclic generating sequence} of $\langle X \rangle$-- see \cite[Chapter~8]{Holt2005HandbookOC}). Note that $m$ is bounded by $\log n$ since $\langle c_i, \dots, c_m\rangle$ is a proper subgroup of $\langle c_{i-1}, \dots, c_m\rangle$ for all $i$.

Next, let us show that in \NTISPpll we can non-deterministically guess elements such that (1) all guessed elements are in $\langle X \rangle$ and (2) every power-product of elements of $C'$ occurs on at least one non-deterministic branch. 

This is not difficult: we start by guessing $x_1$ and set $a=x_1$. Then, guess whether or not to continue; as long as we guess to continue, we guess $i$ and compute $a := [a,x_i]$. As $\ell$ (in the definition of $C$) is bounded by $\log n$, we can guess any element of $C'$ in this way in polylogarithmic time. Moreover, note that at any point we only need to store a constant number of group elements. By construction, all such guessed elements are visibly in $\langle X \rangle$.

Once we have guessed an element of $C'$, we can guess a power of it as in  \cref{obs:logpowerbasis} and then guess the next element of $C'$ and so on.
As $m\leq \log n$, this gives an \NTISPpll\xspace algorithm.
\end{proof}

\section{Abelian Group Isomorphism}
In this section, we consider isomorphism testing of Abelian groups. Our main result in this regard is:

\begin{theorem} \label{thm:abelian}
Let $G$ be an Abelian group, and let $H$ be arbitrary. We can decide isomorphism between $G$ and $H$ in $\forall^{\log \log n} \mathsf{MAC}^0(\DTISPpll)$.
\end{theorem}

Chattopadhyay, Torán, and Wagner \cite{ChattopadhyayToranWagner} established a $\mathsf{TC}^{0}(\FOLL)$ upper bound on this problem. Grochow \& Levet \cite[Theorem 5]{GrochowLevetWL} gave a tighter analysis of their algorithm, placing it in the sub-class $\alphacc{1}\MACz(\FOLL)$.\footnote{Grochow \& Levet consider $\alphacc{1}\MACz \circ \FOLL$, where $\circ$ denotes composition (see \cite{GrochowLevetWL} for a precise formulation). We note that as $\ACz \circ \FOLL = \FOLL = \ACz(\FOLL)$, we have that $\alphacc{1}\MACz \circ \FOLL = \alphacc{1}\MACz(\FOLL)$. Thus, \Thm{thm:abelian} improves upon the previous bound of $\alphacc{1}\MACz(\FOLL)$ obtained by Grochow \& Levet.} 
We note that Chattopadhyay, Tor\'an, \& Wagner also established an upper bound of $\LogSpace$ for this problem, which is incomparable to the result of Grochow \& Levet (\emph{ibid.}). We improve  upon both these bounds by (i) showing that $O(\log \log n)$ non-deterministic bits suffice instead of $O(\log n)$ bits, and (ii) using an $\ACz(\DTISPpll)$ circuit for order finding rather than an $\FOLL$ circuit. We note that while $\alphacc{1}\MACz(\FOLL)$ is contained in $\TCz(\FOLL)$, it is open whether this containment is strict. In contrast, \Cor{cor:MAC} shows that our new bound of $\forall^{\log \log n} \mathsf{MAC}^0(\DTISPpll)$ is a class that is in fact \textit{strictly} contained in $\LogSpace \cap \mathsf{TC}^{0}(\FOLL)$.

Jeřábek also previously established bounds of $\Sigma_{2}$-$\mathsf{TIME}(\log^{2} n)$ for isomorphism testing of Abelian groups, which can be simulated by depth-$3$ $\qACz$-circuits of size $n^{O(\log n)}$ \cite{JerabekTCSStackExchange}; Jeřábek's result and our Theorem~\ref{thm:abelian} are incomparable.

\begin{proof}[Proof of \Thm{thm:abelian}]
Following the strategy of \cite[Theorem~7.15]{GrochowLevetWL}, we show that being \textit{non}-isomorphic can be decided in the same class but with existentially quantified non-deterministic bits. 

We may check in \ACz whether a group is Abelian. So if $H$ is not Abelian, we can decide in \ACz that $G \not \cong H$. So suppose now that $H$ is Abelian. By the Fundamental Theorem of Finite Abelian Groups, $G$ and $H$ are isomorphic if and only if their multisets of orders are the same. In particular, if $G \not \cong H$, then there exists a prime power $p^{e}$ such that there are more elements of order $p^{e}$ in $G$ than in $H$. We first identify the order of each element, which is $\ACz(\DTISPpll)$-computable by \Lem{lem:fastexp}.

We will show how to nondeterministically guess and check the prime power $p^e$ such that $G$ has more elements of order $p^e$ than $H$ does. Let $n = p_{1}^{e_{1}} \cdots p_{\ell}^{e_{\ell}}$ be the prime factorization of $n$. We have that $\ell \leq \log_{2} n$, and the number of distinct prime powers dividing $n$ is $e_{1} + \cdots + e_{\ell} \leq \log_{2} n$. Nondeterministically guess a pair $(i,e)$ with $1 \leq i \leq \lfloor \log_2 n \rfloor$ and $1 \leq e \leq \lfloor \log_2 n \rfloor$. We treat this pair as representing the prime power $p_i^e$ at which we will test that $G$ has more elements of that order than $H$. Because both $i$ and $e$ are bounded in magnitude by $\log_2 n$, the number of bits guessed is at most $\log_2 \log_2 n$. So we may effectively guess $p^{e}$ using $O(\log \log n)$ bits to specify $p$ (implicitly, i.e., by its index $i$) and to specify $e$ (explicitly, i.e., in its binary expansion). 

However, to identify elements of order $p^e$, we will need the number $p^e$ explicitly, in its full binary expansion. First we show that we can get $p^e$ explicitly if we can get $p$ explicitly. Once we have $p$ in its binary expansion, the function $(p,e) \mapsto p^e$ can be computed in $\mathsf{AC}^0$, by \Thm{lem:multiplication}. 
Thus all that remains is to get $p$ explicitly.

By Lemma~\ref{lem:factor} we can compute in \ACz the prime factors of $|G|$. Consider this list of primes as a list of length $O(\log n)$, consisting of numbers each of $O(\log n)$ bits. Now for each pair of primes $p_{j}, p_{h}$, we define an indicator $X(j,h) = 1 \iff p_{j} > p_{h}$. As $p_{j}, p_{h}$ are representable using $O(\log n)$ bits, we may compute $X(j, \ell)$ in $\mathsf{AC}^{0}$. Now as the number of primes $\ell \leq \log_{2}(n)$, we may in $\ACz$ find a prime $j$ with:
\[
\sum_{h=1}^{\ell} X(j,h) = i.
\]

\noindent The result now follows. 
\end{proof}

As with many of our other results, we show that this class is restrictive enough that it cannot compute \algprobm{Parity}. To do this, we appeal to the following theorem of Barrington \& Straubing:

\begin{theorem}[{Barrington \& Straubing \cite[Thm.~7]{BS}}] \label{thm:BS}
Let $k > 1$. Any $\mathsf{TC}$ circuit family of constant depth, size $2^{n^{o(1)}}$, and with at most $n^{o(1)}$ $\mathsf{Majority}$ gates cannot compute the $\algprobm{Mod}_k$ function.
\end{theorem}

\begin{corollary} \label{cor:MAC}
Let $k > 1$, and let $Q^{o(\log n)}$ be any finite sequence (of $O(1)$ length) of alternating $\exists$ and $\forall$ quantifiers, where the total number of bits quantified over is $o(\log n)$. Then 
\[
\algprobm{Mod}_k \notin Q^{o(\log n)} \mathsf{MAC}^0(\DTISPpll).
\]
\end{corollary}

\begin{proof}
Let $L \in Q^{o(\log n)} \mathsf{MAC}^0(\DTISPpll)$. 
Since by Fact~\ref{fact:NTISP} we know that $\DTISPpll \subseteq \qACz$, we have $\mathsf{MAC}^0(\DTISPpll) \subseteq \mathsf{quasiMAC^0}$, that is, quasi-polynomial size circuits of constant depth with a single $\mathsf{Majority}$ gate at the output. Thus $L \in Q^{o(\log n)} \mathsf{quasiMAC}^0$.

Let $C$ be the $\mathsf{quasiMAC}^0$ circuit such that $x \in L \iff (Qy) C(x,y)$, where $|y| < o(\log |x|)$, for all strings $x$. There are $2^{o(\log n)}=n^{o(1)}$ possible choices for $y$; let $C_y(x) = C(x,y)$, where $C_y$ denotes the circuit $C$ with the second inputs fixed to the string $y$.

Now, to compute $L$, the quantifiers $Q^{o(\log n)}$ can be replaced by a constant-depth circuit (whose depth is equal to the number of quantifier alternations), and whose total size is $2^{o(\log n)} = n^{o(1)}$, where at each leaf of this circuit, we put the corresponding circuit $C_y$. The resulting circuit is a $\mathsf{quasiTC}^0$ circuit with only $n^{o(1)}$ $\mathsf{Majority}$ gates, hence by Theorem~\ref{thm:BS}, cannot compute \algprobm{Parity}.
\end{proof}

\section{Isomorphism for Groups of Almost All Orders}

Dietrich \& Wilson \cite{DietrichWilson} previously established that there exists a dense set $\Upsilon \subseteq \mathbb{N}$ such that if $n \in \Upsilon$ and $G_{1}, G_{2}$ are magmas of order $n$ given by their multiplication tables, we can (i) decide if $G_{1}, G_{2}$ are groups, and (ii) if so, decide whether $G_{1} \cong G_{2}$ in time $O(n^{2} \log^{2} n)$, which is quasi-linear time relative to the size of the multiplication table.

In this section, we establish the following.

\begin{theorem} \label{thm:ParallelDW}
Let $\Upsilon \subseteq \mathbb{N}$ be the dense set considered by Dietrich and Wilson \cite{DietrichWilson}. Let $n \in \Upsilon$, and let $G_{1}, G_{2}$ be groups of order $n$. We can decide isomorphism between $G_{1}$ and $G_{2}$ in $\ACz(\DTISPpll)$.   
\end{theorem}

Note that verifying the group axioms is $\ACz$-computable. 

\begin{remark}
Theorem~\ref{thm:ParallelDW} provides that for almost all orders, \algprobm{Group Isomorphism} belongs to \\$\ACz(\DTISPpll)$, which is contained within $\LogSpace \cap \FOLL \subsetneq \mathsf{P}$ and cannot compute \algprobm{Parity}. While it is known that \algprobm{Group Isomorphism} belongs to complexity classes such as $\betacc{2}\LogSpace \cap \betacc{2}\mathsf{FOLL}$ \cite{ChattopadhyayToranWagner} and $\qACz$ (Theorem~\ref{thm:QuasigroupIsoPolylogtime}) that cannot compute \algprobm{Parity}, membership within $\mathsf{P}$---let alone a subclass of $\mathsf{P}$ that cannot compute \algprobm{Parity}---is a longstanding open problem.
\end{remark}

\begin{proof}[Proof of Theorem~\ref{thm:ParallelDW}]
Dietrich \& Wilson showed \cite[Theorem~2.5]{DietrichWilson} that if $G$ is a group of order $n \in \Upsilon$, then $G = H \ltimes B$, where $\gcd(|B|,|H|)=1$ and:
\begin{itemize}
    \item $B$ is a cyclic group of order $p_{1} \cdots p_{\ell}$, where for each $i \in [\ell]$, $p_{i} > \log \log n$ and $p_{i}$ is the maximum power of $p_{i}$ dividing $n$.
    \item $|H| = (\log n)^{\poly \log \log n}$; and in particular, if a prime divisor $p$ of $n$ satisfies $p \leq \log \log n$, then $p$ divides $|H|$. 
\end{itemize} 

As $G_{1}, G_{2}$ are given by their multiplication tables, we may in $\ACz$ compute (i) the prime divisors $p_{1}, \ldots, p_{k}$ of $n$, and (ii) determine whether, for each $i \in[k]$, $p_{i}$ is the maximal power of $p_{i}$ dividing $n$. Furthermore, in $\AC^{0}$, we may write down $\lfloor \log \log n \rfloor$ and test whether $p_{i} > \lfloor \log \log n \rfloor$.

Fix a group $G$ of order $n$. We will first discuss how to decompose $G = H \ltimes B$, as prescribed by \cite[Theorem~2.5]{DietrichWilson}.
Without loss of generality, suppose that $p_{1}, \ldots, p_{\ell}$ ($\ell \leq k$) are the unique primes where $p_{i}$ ($i \in [\ell]$) divides $n$ only once and $p_{i} > \log \log n$. Now as each $p_{i}$ can be represented as a string of length $\leq \lceil \log(n)\rceil+1$, we may in $\AC^{0}$ compute $\overline{p} := p_{1} \cdots p_{\ell}$ (\Thm{lem:multiplication}). Using \Lem{lem:fastexp}, we may in $\ACz(\DTISPpll)$ identify an element $g \in G$ of order $\overline{p}$.

Now in $\ACz$, we may write down the multiplication table for $H_{j} \cong G_{j}/B_{j}$.
As $|H_{j}| \leq (\log n)^{\poly \log \log n}$, there are $\poly(n)$ possible generating sequences for $H_{j}$ of length at most $\log\abs{H_j}$.
By \cite{BabaiSzemeredi} it actually suffices to consider cube generating sequences for $H_{j}$.
Now given cube generating sequences $\overline{x_{j}}$ for $H_{j}$, we may by the proof of Theorem~\ref{thm:QuasigroupIsoPolylogtime} decide whether $\overline{x_{1}} \mapsto \overline{x_{2}}$ extends to an isomorphism of $H_{1}$ and $H_{2}$ in $\alphacc{1}\betacc{1}\DTISPpll \subseteq \ACz(\DTISPpll)$.
As there are only $\poly(n)$ such generating sequences to consider, we may decide whether $H_{1} \cong H_{2}$ in $\ACz(\DTISPpll)$.

Suppose $H_{1} \cong H_{2}$, $B_{1} \cong B_{2}$, and $\text{gcd}(|B_{j}|, |H_{j}|) = 1$ for $j = 1,2$. We have by the Schur--Zassenhaus Theorem that $G_{j} = H_{j} \ltimes_{\theta_{j}} B_{j}$ ($j = 1,2$) for some action $\theta_j \colon H_j \to \text{Aut}(B_j)$. By Taunt's Lemma \cite{Taunt1955}, it remains to test whether the actions $\theta_{1}$ and $\theta_{2}$ are equivalent in the following sense: do there exist  isomorphisms $\alpha \colon H_1 \cong H_2$ and $\beta \colon B_1 \cong B_2$ such that
\[
\beta(h bh^{-1}) = \alpha(h) \beta(b) \alpha(h^{-1}) \qquad \forall h \in H_1, b \in B_1?
\]
Note that, as $B_{j}$ is Abelian, for any two elements $h_{1}, h_{2}$ of $G_{j}$ belonging to the same coset of $B_j$ and any element $b \in B_{j}$, that $h_{1}bh_{1}^{-1} = h_{2}bh_{2}^{-1}$, so the above conjugation action is well-defined, independent of the choice of isomorphic copy of $H_j$ in $G_j$. Next, since the $B_j$ are cyclic, if $b_1$ generates $B_1$, then the above condition is satisfied iff for all $h \in H_1$ we have
\[
\beta(h b_1 h^{-1}) = \alpha(h) \beta(b_1) \alpha(h^{-1}) \qquad \forall h \in H_1.
\]

As above, in $\ACz(\DTISPpll)$ we can find a generator $b_1 \in B_1$ and a cube generating sequence $h_1,\dotsc,h_k \in H_1$. In parallel, for all $\poly(n)$ generators $b_2 \in B_2$ and all $\poly(n)$ cube generating sequences $h'_1,\dotsc,h'_k \in H_2$, we then test the above condition on the isomorphisms specified by $\beta(b_1) = b_2$ and $\alpha(h_i) = h'_i$ for $i=1,\dotsc,k$. That is, for each choice of $b_2$ we write down the isomorphism $B_1 \to B_2$ defined by $\beta(b_1) = b_2$, and for each choice of cube generating sequence $h'_1,\dotsc,h'_k \in H_2$, we check that
\[
\beta(h_i b_1 h_i^{-1}) = h'_i b_2 (h'_i)^{-1} \qquad \forall i=1,\dotsc,k.
\]
Checking the preceding condition only involves a constant number of group multiplications, which can thus be done in $\DTISPpll$.\qedhere
\end{proof}

\section{Quasigroup Isomorphism} \label{sec:QuasigroupIsomorphism}

In this section, we establish the following.

\begin{theorem} \label{thm:QuasigroupIsoPolylogtime}
\algprobm{Quasigroup Isomorphism} belongs to $\betacc{2}\alphacc{1}\betacc{1}\DTISPpll$. In particular, it can be solved by $\qACz$ circuits of depth 4 and size $n^{\Oh(\log n)}$.
\end{theorem}
Note that we have 

\begin{align*}
\betacc{2}\alphacc{1}\betacc{1}\DTISPpll 
&\subseteq\betacc{2}\ACz(\DTISPpll)  \\
&\subseteq\qACz \cap \betacc{2}\LOGSPACE \cap \betacc{2}\FOLL.    
\end{align*}

This statement is inspired by \cite{DietrichEPQW22} where a similar problem is shown to be in the third level of the polynomial-time hierarchy using the same approach. Our proof follows closely the algorithm for \cite[Theorem~3.4]{ChattopadhyayToranWagner}.

\begin{proof}[Proof of Theorem \ref{thm:QuasigroupIsoPolylogtime}]  
Let $G$ and $H$ be quasigroups given as their multiplication tables. We assume that the elements of the quasigroups are indexed by integers $1 , \dots, \abs{G}$. If $\abs{G} \neq \abs{H}$ (this can be tested in $\ComplexityClass{DTIME}(\log n)$ by a standard binary search), we know that $G$ and $H$ are not isomorphic. 
Otherwise, let us write $n= \abs{G}$  and $k = \ceil{2 \log n} + 1$ (while \cite[Theorem 3.3]{ChattopadhyayToranWagner} only states that some $k \in O(\log n)$ suffices, the corresponding proof states that for $k = \ceil{2 \log n} + 1$, there is actually a cube generating set).

 The basic idea is to guess cube generating sequences $(g_0, \dots, g_k)$ and $(h_0, \dots, h_k)$ for $G$ and $H$ and verify that the map $g_i \mapsto h_i$ induces an isomorphism between $G$ and $H$. Hence, we start by guessing cube generating sequences $(g_0, \dots, g_k)$ and $(h_0, \dots, h_k)$ with respect to the parenthesization $P$ where the elements are evaluated left-to-right (so $P(g_{1}g_{2}g_{3}) = (g_{1}g_{2})g_{3}$), with $g_i \in G$, $h_i \in H$.
 This amounts to guessing $2k\cdot \log(n) \in \Oh(\log^2 n)$ many bits (thus, $\betacc{2}$).
Now, we need to verify two points in $\alphacc{1}\betacc{1}\DTISPpll$:
\begin{itemize}
    \item that these sequences are actually cube generating sequences,
    \item that $g_i \mapsto h_i$ induces an isomorphism.
\end{itemize}

Let us describe the first point for $G$ (for $H$ this follows exactly the same way): we universally verify for every element $g \in G$ (which can be encoded using $\Oh(\log n)$ bits, hence, $\alphacc{1}$) that we can existentially guess a sequence $(e_1, \dots, e_k) \in \{0,1\}^k$ (i.e.\ $\betacc{1}$) such that $g = P(g_0g_1^{e_1} \cdots g_k^{e_k})$. 
We can compute this product in $\DTISP(\log^{2+o(1)} n, \log n)$ by multiplying from left to right:\footnote{Here, we could impose the additional requirement that the length of each row/column of the multiplication table is padded up to a power of two in order to get a bound of $\betacc{2}\alphacc{1}\betacc{1}\DTISP(\log^2n,\log n)$).\label{fn:16}} Each multiplication can be done in time $\Oh(\log^{1+o(1)} n)$ because we simply need to compute $i+j\cdot n$ for two addresses $i,j$ of quasigroup elements, write the result on the index tape and then read the corresponding product of group elements from the multiplication table. Moreover, note that for this procedure we only need to store one intermediate result on the working tape at any time
, and one quasigroup multiplication only queries $O(\log n)$ bits (this point will be important for our size analysis below); the ``$+o(1)$'' in the exponent is merely to perform the arithmetic. 
Thus, computing the product $P(g_0g_1^{e_1} \cdots g_k^{e_k})$ can be done in time $\Oh(\log^{2+o(1)} n)$ and space $\Oh(\log n)$ and it can be checked whether the result is $g$.

To check the second point, by \cite{ChattopadhyayToranWagner}, we need to verify universally that for all $(c_1, \dots, c_k)$, $(d_1, \dots, d_k)$, $(e_1, \dots, e_k) \in \{0,1\}^k$ (hence, $\alphacc{1}$) whether 
\begin{align*}
& P(g_0g_1^{c_1} \cdots g_k^{c_k})\cdot P(g_0g_1^{d_1} \cdots g_k^{d_k}) = P(g_0g_1^{e_1} \cdots g_k^{e_k}) \\
 \iff & P(h_0h_1^{c_1} \cdots h_k^{c_k})\cdot P(h_0h_1^{d_1} \cdots h_k^{d_k}) = P(h_0h_1^{e_1} \cdots h_k^{e_k})
\end{align*}
These products can be computed in the same asymptotic complexity bounds as the product above, using the same technique. Now, it remains to observe that we can combine the two $\alphacc{1}$ blocks into one  $\alphacc{1}$ block because the check whether $g_i \mapsto h_i$ induces an isomorphism does not depend on the existentially guessed bits in the first check.

The depth 4 circuit bound then follows from the DNF for $\DTIME(\text{polylog}(n))$ (Fact~\ref{fact:NTISP}): the quantifier blocks $\betacc{2}\alphacc{1}\betacc{1}$ turn into $\bigvee^{\log^2 n} \bigwedge^{\log n} \bigvee^{\log n}$, and then the disjunction at the top of the DNF for $\DTIME(\text{polylog}(n))$ can be merged into the final $\bigvee$ from the quantifier blocks. The size bound is computed as follows: the $\bigvee^{\log^2 n} \bigwedge^{\log n} \bigvee^{\log n}$ multiplies the size by $2^{\Oh(\log^2 n + 2 \log n)}$. For each of the above two points, we showed it can be decided in $\mathsf{DTIME}((\log n)^{2 + o(1)})$ using at most $\Oh(\log^2 n)$ queries to the input. By the proof of Fact~\ref{fact:NTISP}, those yield decision trees of depth $\Oh((\log n)^2)$, giving a DNF of size at most $n \cdot 2^{\Oh(\log n)^2}$. Multiplying all these factors together gives size $2^{\Oh(\log^2 n)} = n^{\Oh(\log n)}$.
\end{proof}

\begin{corollary} \label{cor:SRGs}
The following problems belong to $\betacc{2}\ACz(\DTISPpll)$.
\begin{enumerate}[label=(\alph*)]
\item \algprobm{Latin Square Isotopy}. In particular, it belongs to \\ $\betacc{2}\alphacc{1}\betacc{1}\DTISPpll$, which yields depth $4$ $\qACz$ circuits.

\item Isomorphism testing of Steiner triple systems. In particular, this problem belongs to \\ $\betacc{2}\alphacc{1}\betacc{1}\NTISPpll$, which yields depth $4$ $\qACz$ circuits.

\item Isomorphism testing of Latin square graphs.
\item Isomorphism testing of Steiner $(t, t+1)$-designs.
\item Isomorphism testing of pseudo-STS graphs.
\end{enumerate}
\end{corollary}

\noindent For \algprobm{Latin Square Isotopy} and isomorphism testing of Steiner triple systems, a careful analysis yields depth-$4$ $\qACz$ circuits. In contrast, the reductions from \cite{LevetLatinSquares} for isomorphism testing of Latin square graphs, Steiner $(t, t+1)$-designs, and pseudo-STS graphs all use circuits of depth at least $4$. Obtaining reductions of depth less than $4$ will likely require new techniques.

\begin{proof}
We proceed as follows.
\begin{enumerate}[label=(\alph*)]
\item We note that the reductions outlined in \cite[Theorem~2]{MillerTarjan} and \cite[Remark~1.6]{LevetLatinSquares} in fact allow us to determine whether two quasigroups are isotopic, and not just main class isomorphic. Thus, we have an $\ACz$-computable disjunctive truth-table reduction from \algprobm{Latin Square Isotopy} to \algprobm{Quasigroup Isomorphism}. This suffices to yield the bound of $\betacc{2}\ACz(\DTISPpll)$. 

We now turn to establishing the stronger $\betacc{2}\alphacc{1}\betacc{1}\DTISPpll$ bound. We carefully analyze the $\betacc{2}\LogSpace \cap \betacc{2}\FOLL$ algorithm from \cite[Section~3]{LevetLatinSquares} using cube generating sequences. Our analysis adapts the proof technique of Theorem~\ref{thm:QuasigroupIsoPolylogtime}. Let $Q_1, Q_2$ be quasigroups, and let $k \in O(\log n)$. We begin by guessing three cube generating sequences $\bar{a} = (a_0, a_1, \ldots, a_k), \bar{b} = (b_0, b_1, \ldots, b_k), \bar{c} = (c_0, c_1, \ldots, c_k)$ for $Q_1$, and $\bar{a'} = (a_{0}', a_{1}', \ldots, a_{k}'), \bar{b'} = (b_{0}', b_{1}', \ldots, b_{k}'),\bar{c'} = (c_{0}', c_{1}', \ldots, c_{k}')$ for $Q_2$. This requires $O(\log^2 n)$ non-deterministic bits ($\betacc{2}$). From the proof of Theorem~\ref{thm:QuasigroupIsoPolylogtime}, we may check in $\alphacc{1}\betacc{1}\DTISPpll$ whether each of $\bar{a}, \bar{b}, \bar{c}$ generates $Q_1$, and whether $\bar{a'}, \bar{b'}, \bar{c'}$ generates $Q_2$. 

So now suppose that each of $\bar{a}, \bar{b}, \bar{c}$ generates $Q_1$, and each of $\bar{a'}, \bar{b'}, \bar{c'}$ generates $Q_2$. Let $\alpha, \beta, \gamma : Q_1 \to Q_2$ be the bijections induced by the respective maps $\bar{a} \mapsto \bar{a'}, \bar{b} \mapsto \bar{b'}$, and $\bar{c} \mapsto \bar{c'}$. Now for each pair $(g, h) \in Q_1$, there exist $x, y, z \in \{0,1\}^{k}$ s.t.:
\begin{align*}
&g = a_{0}a_{1}^{x_{1}} \cdots a_{k}^{x_{k}}, \\
&h = b_{0}b_{1}^{y_{1}} \cdots b_{k}^{y_{k}} \\
&gh = c_{0}c_{1}^{z_{1}} \cdots c_{k}^{z_{k}}.
\end{align*}

We require $O(\log n)$ universally quantified co-nondeterministic bits ($\alphacc{1}$) to represent $(g, h)$, and $O(\log n)$ existentially quantified non-deterministic bits to represent $x, y, z$ ($\betacc{1}$). Given $x, y, z$, we may write down $g, h, gh$ in $\DTISPpll$ (following the proof of Theorem~\ref{thm:QuasigroupIsoPolylogtime}). Similarly, let:
\begin{align*}
&g' = a_{0}'(a_{1}')^{x_{1}} \cdots (a_{k}')^{x_{k}}, \\
&h' = b_{0}'(b_{1}')^{y_{1}} \cdots (b_{k}')^{y_{k}} \\
&\ell' = c_{0}(c_{1}')^{z_{1}} \cdots (c_{k}')^{z_{k}}.
\end{align*}

We can check in $\DTISPpll$ whether $g' h' = \ell'$. Similar to the proof of Theorem~\ref{thm:QuasigroupIsoPolylogtime}, we can combine the two $\alphacc{1}$ blocks into one $\alphacc{1}$ block, as the check whether $(\bar{a}, \bar{b}, \bar{c}) \mapsto (\bar{a'}, \bar{b'}, \bar{c'})$ preserves the quasigroup operation (the \emph{homotopism} condition) is independent of whether each of $\bar{a}, \bar{b}, \bar{c}$ generates $Q_1$, and each of $\bar{a'}, \bar{b'}, \bar{c'}$ generates $Q_2$. 

By Fact~\ref{fact:NTISP}, an $\NTISPpll$ machine can be simulated by a depth-$2$ $\qACz$ circuit, taking care of the final $\betacc{1}\DTISPpll$. The additional quantifiers $\betacc{2}\alphacc{1}$ then yield depth-$4$ $\qACz$ circuits. 

\item Given a Steiner triple system, we may obtain a quasigroup $Q$ in the following manner. For each block $\{a,b,c\}$ in the Steiner triple system, we include the products $ab = c, ba = c, ac = b, ca = b, bc = a, cb = a$. We can write down the multiplication table using an $\ACz$ circuit, which suffices to obtain a bound of $\betacc{2}\ACz(\DTISPpll)$. 

However, we can use the blocks of the Steiner triple system to look up the relevant products in $\NTISPpll$, and so we need not write down the multiplication table. This allows us to directly apply the proof of Theorem~\ref{thm:QuasigroupIsoPolylogtime}, which yields a bound of 
\[
\betacc{2}\alphacc{1}\betacc{1}\NTISPpll,
\]
and so we obtain depth-$4$ $\qACz$ circuits, as desired. (Analyzing the $\ACz$ reduction in terms of circuits, yielded only depth 6 in the end, whereas realizing the reduction can be computed in $\NTISPpll$ allows us to get depth 4 overall.)

\item Given two Latin square graphs $G_{1}, G_{2}$, we may recover corresponding Latin squares $L_{1}, L_{2}$ in $\ACz$ (cf. \cite[Lemma~3.9]{LevetLatinSquares}). Now $G_{1} \cong G_{2}$ if and only if $L_{1}$ and $L_{2}$ are main class isomorphic \cite[Lemma~3]{MillerTarjan}. We may decide whether $L_{1}, L_{2}$ are main class isomorphic using an $\ACz$-computable disjunctive truth-table reduction to \algprobm{Quasigroup Isomorphism} (cf. \cite[Remark~1.6]{LevetLatinSquares}). By Fact~\ref{fact:dtt},  $\betacc{2}\ACz(\DTISPpll)$ is closed under $\ACz$-computable dtt reductions, thus yielding the bound of \\$\betacc{2}\ACz(\DTISPpll)$.

\item This immediately follows from the fact that  isomorphism testing of Steiner $(t,t+1)$-designs is $\betacc{2}\ACz$-reducible to isomorphism testing of Steiner triple systems \cite{BabaiWilmes} (cf., \cite[Corollary~4.11]{LevetLatinSquares}).

\item Bose \cite{Bose} previously showed that pseudo-STS graphs with $> 67$ vertices are STS graphs. Now given a block-incidence graph from a Steiner $2$-design with bounded block size, we can recover the underlying design in \ACz \cite[Proposition~4.7]{LevetLatinSquares}. This yields the desired bound.\qedhere
\end{enumerate}
\end{proof}

\begin{remark} \label{rmk:LatinSquares}
 Latin square graphs are one of the four families of strongly regular graphs under Neumaier's classification \cite{Neumaier} (the other families being line graphs of Steiner $2$-designs, conference graphs, and graphs whose eigenvalues satisfy the claw bound). Levet \cite{LevetLatinSquares} previously established an upper bound of $\betacc{2}\ACz$ for isomorphism testing of conference graphs, which is a stronger upper bound than we obtain for Latin square graphs. In contrast, the best known algorithmic runtime for isomorphism testing of conference graphs is $n^{2\log(n) + O(1)}$ due to Babai \cite{BabaiCanonicalLabeling1}, whereas isomorphism testing of Latin square graphs is known to admit an $n^{\log(n) + O(1)}$-time solution \cite{MillerTarjan}.
\end{remark}

\section{Minimum Generating Set}

In this section, we consider the \algprobm{Minimum Generating Set} (\algprobm{MGS}) problem for quasigroups, as well as arbitrary magmas.

\subsection{MGS for Groups in \texorpdfstring{$\AC^1(\LogSpace)$}{AC1(LogSpace)}}

In this section, we establish the following.

\begin{theorem} \label{thm:MGSSAC2}
\algprobm{MGS} for groups belongs to $\AC^{1}(\LogSpace)$.
\end{theorem}

We begin with the following lemma.

\begin{lemma} \label{lem:ChiefSeriesSAC2}
Let $G$ be a group. We can compute a chief series for $G$ in $\AC^{1}(\LogSpace)$.
\end{lemma}

\begin{proof}
We will first show how to compute the minimal normal subgroups $N_{1}, \ldots, N_{\ell}$ of $G$. We proceed as follows. We first note that the normal closure $\text{ncl}(x)$ is the subgroup generated by $\{ gxg^{-1} : g \in G \}$. Now we may write down the elements of $\{ gxg^{-1} : g \in G \}$ in $\ACz$, and then compute $\text{ncl}(x)$ in $\LogSpace$ using a membership test. Now in $\LogSpace$, we may identify the minimal (with respect to inclusion) subgroups amongst those obtained. 

Given $N_{1}, \ldots, N_{\ell}$, we may easily in $\LogSpace$ compute $\prod_{i=1}^{k} N_{i}$ for each $k \leq \ell$. In particular, we may compute $\Soc(G)$ in $\LogSpace$. 
We claim that $\prod_{i=1}^k N_i$, for $k=1,\dotsc,\ell$, is in fact a chief series of $\Soc(G)$ (which will then fit into a chief series for $G$). To see this, we have that $\prod _{i=1} ^k N_i$ is normal in $\left( \prod _{i=1} ^k N_i \right) \times N_{k+1}$ and that $(\prod _{i=1} ^{k+1} N_i ) / (\prod _{i=1} ^{k} N_i ) \cong N_{k+1}$ is a normal subgroup of $G / \prod_{i=1} ^k N_i $. By the Lattice Isomorphism Theorem, $(\prod _{i=1} ^{k+1} N_i ) / (\prod _{i=1} ^{k} N_i )$ is in fact minimal normal in $G / \prod_{i=1} ^k N_i $.

We iterate on this process starting from $G/\Soc(G)$. Note that, as we have computed  $\Soc(G)$ from the previous paragraph, we may write down the cosets for $G/\Soc(G)$ in $\ACz$. Furthermore, given a subgroup $H \leq G/\Soc(G)$, we may write down the elements of  $H\Soc(G)$ in $\ACz$. By the above, the minimal normal subgroups of a group are computable in $\LogSpace$. 
As there are at most $\log n$ terms in a chief series, we may compute a chief series for $G$ in $\AC^{1}(\LogSpace)$, as desired.  
(Recall that we use this notation to mean an $\AC^1$ circuit with oracle gates calling a $\LogSpace$ oracle, not function composition such as $\AC^1 \circ \LogSpace$.)
\end{proof}

We now prove the Theorem~\ref{thm:MGSSAC2}.

\begin{proof}[Proof of Theorem~\ref{thm:MGSSAC2}]
By \Lem{lem:ChiefSeriesSAC2}, we can compute a chief series for $G$ in $\AC^{1}(\LogSpace)$. So let $N_1 \triangleleft \cdots \triangleleft N_k$ be a chief series $S$ of $G$.  Lucchini and Thakkar \cite{LucchiniThakkar} showed that minimum generating sets of $G / N_{i+1}$ have specific structure depending on whether or not $N_{i+1}/N_{i}$ is Abelian. We proceed inductively down $S$ starting from $N_{k-1}$. As $G / N_{k-1}$ is a finite simple group, and hence at most $2$-generated, we can write all $\binom{n}{2}$ possible generating sets in parallel with a single $\AC^0$ circuit and test whether each generates the group with a membership test. This can be done in $\LogSpace$. 

Fix $i < k$. Suppose we are given a minimum generating sequence $g_{1}, \ldots, g_{d} \in G$ for $G/N_{i}$. We will construct a minimum generating sequence for $G/N_{i-1}$ as follows. We consider the following cases:
\begin{itemize}
\item \textbf{Case 1:} Suppose that $N = N_{i}/N_{i-1}$ is Abelian. By \cite[Theorem~4]{LucchiniMenegazzo}, we have two cases: 
\begin{itemize}

\item \textbf{Case 1a:} We have $G/N_{i-1} = \langle g_1, \cdots, g_i, g_j n, g_{j+1}, \cdots, g_d \rangle$ for some $j \in [d]$ and some $n \in N$ (possibly $n = 1$). There are at most $d \cdot |N|$ generating sets to consider in this case and we can test each of them in $\LogSpace$. 

\item \textbf{Case 1b:} If Case 1a does not hold, then we necessarily have that \\ $G/N_{i-1} = \langle g_1, \cdots, g_d, x \rangle$ for any non-identity element $x \in N$.
\end{itemize}

Note that there are at most $d \cdot |N| + 1$ generating sets to consider, we may construct a minimum generating set for $G/N_{i-1}$ in $\LogSpace$ using a membership test.

\item \textbf{Case 2:} Suppose instead that $N = N_{i}/N_{i-1}$ is non-Abelian. Then by \cite[Corollary 13]{LucchiniThakkar}, the following holds. Let $\eta_{G}(N)$ denote the number of factors in a chief series with order $|N|$. Let $u = \max\{d, 2\}$ and $t = \min\{u, \lceil \frac{8}{5} + \log_{\lvert N \rvert}\eta_G(N) \rceil\}$. Then there exist $n_{1}, \ldots, n_{t} \in N_{i-1}$ (possibly $n_1 = \cdots = n_t = 1$) such that $G/N_{i-1} = \langle g_1n_1, \cdots, g_tn_t, g_{t+1}, \cdots, g_d \rangle$.

By \cite[Corollary 13]{LucchiniThakkar}, there are at most $\lvert N \rvert ^{\lceil \frac{8}{5} + \log_{\lvert N \rvert}\eta_G(N)  \rceil }$ generating sets of this form. As $\log_{|N|} \eta_{G}(N) \in O(1)$, we may write down these generating sets in parallel with a single $\AC^0$ circuit and test whether each generate $G/N_{i-1}$ in $\LogSpace$ using \algprobm{Membership}.
\end{itemize}

\noindent Descending along the chief series in this fashion, we compute quotients $N_i / N_{i-1}$ and compute a generating set for $G/N_{i-1}$. The algorithm terminates when we've computed a generating set for $G/N_0 = G$. Since a chief series has $O(\log n)$ terms, this algorithm requires $O(\log n)$ iterations and each iteration is computable in $\LogSpace$. Hence, we have an algorithm for \algprobm{MGS} in $\AC^1(\LogSpace)$.
\end{proof}

Improving upon the $\AC^{1}(\LogSpace)$ bound on \algprobm{MGS} for groups appears daunting. It is thus natural to inquire as to families of groups where \algprobm{MGS} is solvable in complexity classes contained within $\AC^{1}(\LogSpace)$. To this end, we examine the class of nilpotent groups. Arvind \& Tor\'an previously established a polynomial-time algorithm for nilpotent groups \cite[Theorem~7]{ArvindToran}. We improve their bound as follows.

\begin{proposition} \label{prop:Nilpotent}
For a nilpotent group $G$, we can compute $d(G)$ in
\[
\LogSpace \cap \ACz(\NTISPpll).
\]
\end{proposition}

\begin{proof}
Let $G$ be our input group. Recall that a finite nilpotent group is the direct product of its Sylow subgroups (which by the Sylow theorems, implies that for a given prime $p$ dividing $|G|$, the Sylow $p$-subgroup of $G$ is unique). We can, in $\ACz(\DTISPpll)$ (using \Cor{cor:fastexp}), decide if $G$ is nilpotent; and if so, compute its Sylow subgroups. So we write $G = P_{1} \times P_{2} \times \cdots \times P_{\ell}$, where each $P_{i}$ is the Sylow subgroup of $G$ corresponding to the prime $p_{i}$. Arvind \& Tor\'an (see the proof of \cite[Theorem~7]{ArvindToran}) established that $d(G) = \max_{1 \leq i \leq \ell} d(P_{i})$. Thus, it suffices to compute $d(P_{i})$ for each $i \in [\ell]$.

The Burnside Basis Theorem provides that $\Phi(P) = P^{p}[P,P]$. We may compute $P^{p}$ in \\$\LogSpace \cap \ACz(\DTISPpll)$ (the latter using \Cor{cor:fastexp}). We now turn to computing $[P,P]$.
Using a membership test, we can compute $[P,P]$ in $\LogSpace$.
By \cite[I.§4~Exercise~5]{Serre}, every element in $[P,P]$ is the product of at most $\log |P|$ commutators. 
Therefore, we can also decide membership in $[P,P]$ in $\NTISPpll$, and so we can write down the elements of $[P,P]$ in $\ACz(\NTISPpll)$.

Thus, we may compute $\Phi(P)$ in $\LogSpace \cap \ACz(\NTISPpll)$. Given $\Phi(P)$, we may compute $|P/\Phi(P)|$ in $\ACz$. Thus, we may recover $d(P)$ from $|P/\Phi(P)|$ in $\ACz$, by iterated multiplication of the prime divisor $p$ of $|P|$. As the length of the encoding of $p$ is at most $\log |P|$ and we are multiplying $p$ by itself $\log |P|$ times, iterated multiplication is $\AC^{0}$-computable. Thus, in total, we may compute $d(P)$ in $\LogSpace \cap \ACz(\NTISPpll)$. It follows that for an arbitrary nilpotent group $G$, we may compute $d(G)$ in $\LogSpace \cap \ACz(\NTISPpll)$. 
\end{proof}

\begin{remark}
While \Prop{prop:Nilpotent} allows us to compute $d(G)$ for a nilpotent group $G$, the algorithm is non-constructive. It is not clear how to find such a generating set in $\LogSpace$. We can, however, compute such a generating set in  $\AC^{1}(\NTISPpll)$. Note that this bound is incomparable to $\AC^{1}(\LogSpace)$. We outline the algorithm here. 

The Burnside Basis Theorem provides that for a nilpotent group $G$, (i) every generating set of $G$ projects to a generating set of $G/\Phi(G)$, and (ii) for every generating set $S$ of $G/\Phi(G)$, \emph{every} lift of $S$ is a generating set of $G$. Furthermore, every minimum generating set of $G$ can be obtained from the Sylow subgroups in the following manner. Write $G = P_{1} \times \cdots \times P_{\ell}$, where $P_{i}$ is the Sylow $p_{i}$-subgroup of $G$. Suppose that $P_{i} = \langle g_{i1}, \ldots, g_{ik}\rangle$ (where we may have $g_{ij} = 1$ for certain values of $j$). Write $g_{j} = \prod_{i=1}^{\ell} g_{ij}$. As in the proof of \cite[Theorem~7]{ArvindToran} we obtain $G = \langle g_{1}, \ldots, g_{k} \rangle.$

Given generating sets for $P_{1}, \ldots, P_{\ell}$, we may in $\LogSpace \cap \ACz(\NTISPpll)$ recover a generating set for $G$. Thus, it suffices to compute a minimum generating set for $P/\Phi(P) \cong (\mathbb{Z}/p\mathbb{Z})^{d(P)}$, where $P$ is a $p$-group.  Note that we may handle each Sylow subgroup of $G$ in parallel. To compute a minimum generating set of $P/\Phi(P)$, we use the generator enumeration strategy. As $P/\Phi(P)$ is Abelian, we may check in $\ACz(\NTISPpll)$ (by \Cor{cor:cayley_abelian}) whether a set of elements generates the group. As $d(P) \leq \log |P|$, we have $\log |P|$ steps where each step is $\ACz(\NTISPpll)$-computable. Thus, we may compute a minimum generating set for $P$ in $\AC^{1}(\NTISPpll)$, as desired. 
\end{remark}

\subsection{\algprobm{MGS} for Quasigroups}
In this section, we consider the \algprobm{Minimum Generating Set} problem for quasigroups. Our goal is to establish the following.

\begin{theorem} \label{thm:MGSqAC0}
For \algprobm{MGS} for quasigroups, 
\renewcommand{\theenumi}{\alph{enumi}}
\begin{enumerate}
\item \label{thm:MGSqAC0:NTIME} The decision version belongs to $\betacc{2}\alphacc{1}\betacc{1}\DTISPpll \subseteq \mathsf{DSPACE}(\log^2 n)$;

\item \label{thm:MGSqAC0:search} The search version belongs to $\qACz \cap \mathsf{DSPACE}(\log^2 n)$.
\end{enumerate}
\end{theorem}

In the paper in which they introduced (polylog-)limited nondeterminism, Papadimitriou and Yannakakis conjectured that MGS for quasigroups was $\betacc{2}\P$-complete \cite[after Thm.~7]{PY}. While they did not specify the type of reductions used, it may be natural to consider polynomial-time many-one reductions. Theorem~\ref{thm:MGSqAC0} refutes two versions of their conjecture under other kinds of reductions, that are incomparable to polynomial-time many-one reductions: $\qACz$ reductions unconditionally and polylog-space reductions conditionally. We note that their other $\betacc{2}\P$-completeness result in the same section produces a reduction that in fact can be done in logspace and (with a suitable, but natural, encoding of the gates in a circuit) also in $\ACz$, so our result rules out any such reduction for \algprobm{MGS}. 
(We also note: assuming $\EXP \neq \PSPACE$, showing that this problem is complete under polynomial-time reductions would give a separation between poly-time and log-space reductions, an open problem akin to $\P \neq \LogSpace$.)

\renewcommand{\theenumi}{\alph{enumi}}
\begin{corollary}
\algprobm{MGS} for quasigroups and \algprobm{Quasigroup Isomorphism} are 
\begin{enumerate}
\item not $\betacc{2} \P$-complete under $\qACz[p]$ Turing reductions, or even such reductions up to depth $o(\log n / \log \log n)$.

\item not $\betacc{2} \P$-complete under polylog-space Turing reductions unless $\mathsf{EXP} = \mathsf{PSPACE}$.
\end{enumerate}
\end{corollary}

\begin{proof}
(a) \Thm{thm:MGSqAC0}(\ref{thm:MGSqAC0:NTIME}) for \algprobm{MGS}, resp. \Thm{thm:QuasigroupIsoPolylogtime} for \algprobm{Quasigroup Isomorphism}, place both problems into classes that are contained in $\qACz$ by Fact~\ref{fact:NTISP}. If either problem were $\betacc{2}\P$-complete under quasi-polynomial size, $o(\log n / \log \log n)$-depth, $\mathsf{AC}[p]$ circuits, then, since $\algprobm{Parity} \in \P \subseteq \betacc{2}\P$, we would get such circuits for \algprobm{Parity}, which do not exist \cite{Razborov, Smolensky87algebraicmethods}  (Smolensky's argument yields a minimum size of $\exp(\Omega(n^{1/(2d)}))$ for depth-$d$ circuits, which is super-polynomial when $d \in o(\log n / \log \log n)$). 

(b) Both \algprobm{MGS} for quasigroups and \algprobm{Quasigroup Isomorphism} are in $\mathsf{DSPACE}(\log^2 n)$ by \Thm{thm:MGSqAC0}, resp.\ \cite{ChattopadhyayToranWagner}. The closure of $\mathsf{DSPACE}(\log^2 n)$ under poly-log space reductions is contained in $\mathsf{polyL} = \bigcup_{k \geq 0} \mathsf{DSPACE}(\log^k n)$. If either of these two quasigroup problems were complete for $\betacc{2} \P$ under polylog-space Turing reductions, we would get $\betacc{2} \P \subseteq \mathsf{polyL}$. Under the latter assumption, by a straightforward padding argument, we now show that $\EXP = \PSPACE$.

Let $L \in \EXP$; let $k$ be such that $L \in \DTIME(2^{n^k}+k)$. Define $L_{pad} = \{(x,1^{2^{|x|^k}+k}) : x \in L\}$. By construction, $L_{pad} \in \P$. Let us use $N$ to denote the size of the input to $L_{pad}$, that is, $N = 2^{n^k}+k+n$. By assumption, we thus have $L_{pad} \in \mathsf{polyL}$. Suppose $\ell$ is such that $L_{pad} \in \mathsf{DSPACE}(\log^\ell N)$. We now give a $\PSPACE$ algorithm for $L$. In order to stay within polynomial space, we cannot write out the padding $1^{2^{n^k}+k}$ explicitly. What we do instead is simulate the $\mathsf{DSPACE}(\log^\ell N)$ algorithm for $L_{pad}$ as follows. Whenever the head on the input tape would move off the $x$ and into the padding, we keep track of its index into the padding, and the simulation responds \emph{as though} the tape head were reading a 1. When the tape head moves right the index increases by 1, when it moves left it decreases by 1, and if the index is zero and the tape head moves left, then we move the tape head onto the right end of the string $x$. The index itself is a number between $0$ and $2^{n^k}+k$, so can be stored using only $O(n^k)$ bits. The remainder of the $L_{pad}$ algorithm uses only $O(\log^\ell N) = O(n^{k\ell})$ additional space, thus this entire algorithm uses only a polynomial amount of space, so $L \in \PSPACE$, and thus $\EXP=\PSPACE$.
\end{proof}

Now we return to establishing the main result of this section, \Thm{thm:MGSqAC0}. To establish \Thm{thm:MGSqAC0}(\ref{thm:MGSqAC0:NTIME}) and (\ref{thm:MGSqAC0:search}), we will crucially leverage the \algprobm{Membership} for quasigroups problem. To this end, we will first establish the following.
\begin{theorem}\label{thm:CQMqAC0}
\algprobm{Membership} for quasigroups belongs to $\betacc{2}\DTISPpll$.
\end{theorem}

\Thm{thm:CQMqAC0} immediately yields the following corollary. 

\begin{corollary}
For quasigroups, \algprobm{Membership} and \algprobm{MGS} are not hard under $\ACz$-reductions for any complexity class containing \algprobm{Parity}.
\end{corollary}

\paragraph{The Reachability Lemma and proofs.} 
The proofs of \cref{thm:CQMqAC0} and \cref{thm:MGSqAC0} rely crucially on the following adaption of the Babai--Szemer\'edi  Reachability Lemma \cite[Theorem~3.1]{BabaiSzemeredi} to quasigroups. We first generalize the notion of a straight-line program for groups \cite{BabaiSzemeredi} to SLPs for quasigroups. We follow the same strategy as in the proof of \cite[Theorem~3.1]{BabaiSzemeredi}, but there are some subtle, yet crucial, modifications due to the fact that quasigroups are non-associative and need not possess an identity element.

Let $X$ be a set of generators for a quasigroup $G$. We call a sequence of elements $g_{1}, \ldots, g_{\ell} \in G$ a \textit{straight-line program} (SLP for short) if each $g_{i}$ ($i \in [\ell]$) either belongs to $X$, or is of the form or $g_{j}g_{k}$, $g_{j}\backslash g_{k}$, or $g_{j}/ g_{k}$ for some $j, k < i$ (where $g_{j}\backslash g_{k}$, resp. $g_{j}/ g_{k}$, denotes the quasigroup division as defined in \cref{sec:algebraprelims}). An SLP is said to \emph{compute} or \emph{generate} a set $S$ (or an element $g$) if $S \subseteq \{g_1, \dots, g_\ell\}$ (resp.\ $g\in\{g_1, \dots, g_\ell\}$).

For any sequence of elements $z_1,\dotsc,z_k$, let $P(z_1 z_2 \dotsb z_k)$ denote the left-to-right parenthesization, e.g., $P(z_1 z_2 z_3) = (z_1 z_2) z_3$. For some initial segment $z_0, z_{1}, \ldots, z_{i}$ define the cube
\[
K(i) = \{ P(z_0z_{1}^{e_{1}} \cdots z_{i}^{e_{i}}) : e_{1}, \ldots, e_{i} \in \{0,1\}\},
\]
where $e_j=0$ denotes omitting $z_j$ from the product (since there need not be an identity element).
 Define $L(i) = K(i)\backslash K(i) = \{ g\backslash h : g,h \in K(i)\}$.

Note that, if $L(k) = G$, then $z_1,\dotsc,z_k$ is very similar to a cube generating sequence and we call it a  \emph{cube-like} generating sequence (the difference is that $L(k)$ allows quotients of cube words, not only cube words, and in a quasigroup quotients of cube words need not always be cube words). 

\begin{lemma}[Reachability Lemma for quasigroups]\label{extendedReachability} 
Let $G$ be a finite quasigroup and let $X$ be a set of generators for $G$. Then there exists a sequence $z_0, \dots, z_t$ with $t\leq \log \abs{G}$ such that:
\begin{enumerate}[label*=(\arabic*)]
    \item $L(t) = G$
    \item\label{zinL} for each $i$ we have either $z_i \in X$ or $z_i \in \{ gh,\, g/h,\, h\backslash g\mid g,h \in L(i-1) \}$.
\end{enumerate}
In particular, for each $g \in G$, there exists a straight-line program over $X$ generating $g$ which has length $O(\log^2 |G|)$.
\end{lemma}
\begin{proof}
 We will inductively construct the sequence $z_0,z_{1}, \ldots, z_{t}$ as described in the lemma. To start, we take $z_0$ as an arbitrary element from $X$. Hence,  $K(0) = \{z_0\}$. 
Next, let us construct $K(i+1)$ from $K(i)$. If $L(i) \neq G$, we set $z_{i+1}$ to be some  element of $G \setminus L(i)$ that we can find as follows:
As $G \neq L(i)$, either $X \not\subseteq L(i)$ or $L(i)$ is not a sub-quasigroup. Hence, we have one of the following cases:
\begin{itemize}
    \item  If there is some $g \in X \setminus L(i)$, we simply take $z_{i+1} = g$.

    \item Otherwise, $L(i)$ is not a sub-quasigroup; hence, there exist $g,h \in L(i)$ with one of $gh$, $g/h$, or $g\backslash h \not \in L(i)$. We choose $z_{i+1}$ to be one of $gh$, $g/h$, or $g\backslash h$ that is not in $L(i)$.
    
\end{itemize}

Next, we first claim that $|K(i+1)| = 2 \cdot |K(i)|$.
Note that $K(i+1)= K(i) \cup K(i)z_{i+1}$ by definition.
As right-multiplication by a fixed element is a bijection in a quasi-group, it suffices to show that $K(i) \cap K(i)z_{i+1} = \emptyset$.
So, suppose that there exists some $a \in K(i) \cap K(i)z_{i+1}$. Then $a = gz_{i+1}$ for some $g \in K(i)$. Hence, $z_{i+1} = g\backslash a$, contradicting $z_{i+1} \not\in L(i)$, since both $a$ and $g$ are in $K(i)$. 

It now follows that $\abs{K(i)} = 2^i$ and, hence, for some $t \leq \lceil \log_{2}(|G|) \rceil$ we have $L(t) = G$. This completes the proof of the first part of the lemma.

It remains to see that for every $g \in G$ there is a straight-line program over $X$ of length $O(\log^2 \abs{G})$. 
To see this, let $c(i)$ denote the straight-line cost for $\{z_0,z_{1}, \ldots, z_{i}\}$ ($1\leq i \leq t$), which is defined as the length of the shortest SLP generating $\{z_0,z_{1}, \ldots, z_{i}\}$. 
We have $c(0) = 1$ and, if $z_{i+1} \in X \setminus L(i)$, then $c(i+1) \leq c(i) + 1$. 
If $z_{i+1} \in \{ gh,\, g/h,\, h\backslash g\mid g,h \in L(i)\}$, we can write each of $g$ and $h$ as SLP over the set $\{z_0,z_{1}, \ldots, z_{i}\}$ of length at most $2i + 1$, yielding an SLP for $z_{i+1} $ over the same set of length at most $4i+3$. Together this yields $c(i+1) \leq c(i) + 4i+3$.
 Hence, by induction, $c(t) \in O(\log ^2\abs{G})$. As $L(t)= G$, we obtain for any element $g \in G$ an SLP of length $c(t) + 2t+1\in O(\log^2\abs{G})$. 
\end{proof}

\begin{remark}
The proof of Lemma~\ref{extendedReachability} shows that for any quasigroup $Q$ and any generating set $S \subseteq Q$, every element $g \in Q$ can be realized with a parse tree of depth $O(\log^2 |Q|)$. In contrast, Wolf \cite{Wolf} establishes the existence of parse trees of depth $O(\log |Q|)$.   
\end{remark}

For proving \cref{thm:CQMqAC0} we follow essentially the ideas of \cite{Fleischer} (though we avoid introducing the notion of Cayley circuits); however, we have to use a more refined approach by not simply guessing an SLP but the complete information given in \cref{extendedReachability}. Fleischer obtained a $\qACz$ bound for \algprobm{Membership} for groups by then showing that the Cayley circuits for this problem can be simulated by a depth-$2$ $\qACz$ circuit. We will instead directly analyze the straight-line programs using an $\betacc{2}\DTISPpll$ (sequential) algorithm.

\begin{proof}[Proof of Theorem \ref{thm:CQMqAC0}]
    To decide whether $g \in \langle X \rangle$ with $X = \{x_1, \dots, x_\ell\}$, we guess a sequence of elements $z_1, \dots, z_k \in G$ with $k \leq \log n$ as in \cref{extendedReachability} together with the necessary information to verify condition \ref{zinL} in the lemma and whether $g \in L(k)$. This information consists of sequences $\ell_i \in [\ell]$ and $e_{i,j}^{(\mu)}\in \{0,1\}$ for $i \in [k]$, $j\in [i-1]$, $\mu \in [4]$ and, as witness that $g \in L(k)$, a sequence $f_j^{(\mu)}\in \{0,1\}$ for $j\in [k]$, $\mu \in [2]$. Note that all this information can be represented using $O(\log^2n)$ bits.

    In the \DTISPpll computation we compute for all $i$ the elements $g_i^{(\mu)} = P(z_0z_{1}^{e_{i,1}^{(\mu)}} \cdots z_{i-1}^{e_{i,i-1}^{(\mu)}})$ and verify that $z_i $ is either $x_{\ell_i}$ or can be written as $ab$, $a\backslash b$ or $a/b$ where $a = g_i^{(1)} \backslash g_i^{(2)}$ and $b = g_i^{(3)} \backslash g_i^{(4)}$~-- thus, verifying the conditions of \cref{extendedReachability}. Finally, with the same technique we verify that $g\in L(k)$ by checking whether $g = P(z_0z_{1}^{f_{1}^{(1)}} \cdots z_{k}^{f_{k}^{(1)}})  \backslash P(z_0z_{1}^{f_{1}^{(2)}} \cdots z_{k}^{f_{k}^{(2)}}) $.
\end{proof}

\begin{proof}[Proof of Theorem \ref{thm:MGSqAC0}]

(\ref{thm:MGSqAC0:NTIME}) Let $G$ denote the input quasigroup (of order $n$). First, every quasigroup has a generating set of size $\leq \lceil \log n \rceil$ \cite{MillerTarjan}. Therefore, we start by guessing a subset $X\subseteq G$ of size at most $\leq \lceil \log n \rceil$ (resp.\ the size bound given in the input). For this, we use $O(\log^2 n)$ existentially quantified non-deterministic bits ($\betacc{2}$). 
Furthermore, we also guess some additional information, namely, the sequence $z_1, \dots, z_k \in G$ from \cref{extendedReachability} and, just as in the proof of \cref{thm:CQMqAC0}, the witnesses that this is, indeed, a cube-like sequence. As in the proof of \cref{thm:CQMqAC0} this takes again $O(\log^2 n)$ existentially quantified bits.

 In the next step, we verify whether $X$ actually generates $G$. This is done by checking for all $g \in G$ (universally quantifying $O(\log n)$ bits, $\alphacc{1}$) whether $g \in \langle X \rangle$, which can be done in $\betacc{1}\DTISPpll \subseteq \qACz$ just as in \cref{thm:CQMqAC0}. 
 Here it is crucial to note that, as we have already guessed the cube-like generating set in the outer existential block, it suffices to guess $O(\log n)$ bits for the exponents. Finally, as in the proof of \cref{thm:CQMqAC0}, we verify in $\DTISPpll$ using the additional information guessed above that $z_1, \dots, z_k \in G$ is indeed, a cube-like sequence.
 This concludes the proof of the bound (a) for the decision variant. 

(\ref{thm:MGSqAC0:search})
To find a minimum-sized generating sequence,  we can enumerate all possible generating sets $X$ of size at most $\log n$ in $\qACz \cap \mathsf{DSPACE}(\log^2 n)$. 
If we want to compute the minimum generating set, we first, have to find a generating set $X$ and then we have to check whether $X$ is actually of smallest possible size. To do this, in a last step, we use the decision variant from above to check that there is no generating set of size at most $\abs{X} - 1$ for $G$.
\end{proof}

\begin{remark}
While it is possible to directly reduce \algprobm{MGS} to \algprobm{Membership} in $\qACz$, we obtain slightly better bounds by instead directly leveraging the Reachability Lemma~\ref{extendedReachability}. It is also possible to $\qACz$-reduce \algprobm{Quasigroup Isomorphism} to \algprobm{Membership}. This formalizes the intuition that, from the perspective of $\qACz$, membership testing is an essential subroutine for isomorphism testing and \algprobm{MGS}. 
 
This might seem surprising, as in the setting of groups, \algprobm{Membership} belongs to $\LogSpace$, while \algprobm{MGS} belongs to $\mathsf{AC}^{1}(\mathsf{L})$ (\Thm{thm:MGSSAC2}), yet it is a longstanding open problem whether \algprobm{Group Isomorphism} is in $\mathsf{P}$.   
\end{remark}

\subsection{MGS for Magmas}

In this section, we establish the following.

\begin{theorem} \label{thm:magma}
    The decision variant of \algprobm{Minimum Generating Set} for commutative magmas (unital or not) is \NP-complete under many-one $\ACz$ reductions.
\end{theorem}

This \NP-completeness result helps explain the use of Integer Linear Programming in practical heuristic algorithms for the search version of this problem in magmas, e.\,g., \cite{JMV}.

For the closely related problem they called \algprobm{Log Generators}---given the multiplication table of a binary function (=magma) of order $n$, decide whether it has a generating set of size $\leq \lceil \log_2 n \rceil$---Papadimitriou and Yannakakis  proved that \algprobm{Log Generators of Magmas} is $\betacc{2} \mathsf{P}$-complete under polynomial-time reductions \cite[Thm.~7]{PY}. 
Our proof uses a generating set of size roughly $\sqrt{n}$, where $n$ is the order of the magma; this is analogous to the situation that $\algprobm{Log-Clique}$ is in $\betacc{2} \P$, while finding cliques of size $\Theta(\sqrt{n})$ in an $n$-vertex graph is $\NP$-complete.

\begin{proof}
It is straightforward to see that the problem is in \NP by simply guessing a suitable generating set. To show \NP-hardness
we reduce \algprobm{3SAT} to \algprobm{Minimum Generating Set} for commutative magmas. Let $F=\bigwedge_{j=1}^m C_j$ 
with variables $X_1, \dots, X_n$ and clauses $C_1, \dots, C_m$ be an instance of \algprobm{3SAT}. Our magma $M$ consists of the following elements: 
\begin{itemize}
\item for each variable $X_i$, two elements $X_i, \overline{X}_i$, 
\item for each clause $C_j$ an element $C_j$,
\item for each $1\leq j\leq k\leq m$ an element $S_{j,k}$, and 
\item a trash element $0$.
\end{itemize}
We use $S$ as an abbreviation for $S_{1,m}$. 

We define the multiplication as follows: 
\begin{align*}   
 C_j X &= S_{j,j} \text{ if the literal $X$ appears in $C_j$}\\
 S  X_i &= \overline{X}_i,\\ 
 S \overline{X}_i &= X_i, \\
  S_{j,k}S_{k+1,\ell}&= S_{j,\ell}.
\end{align*}
Aside from multiplication being commutative (e.\,g., we also have $X_i S = \overline{X}_i$, etc.), 
all other products are defined as $0$.

The idea is that the presence of $S_{j,k}$ in a word indicates that clauses $j,j+1,\dotsc,k$ have been satisfied. This interpretation aligns with the multiplication above: viz. $C_j X = S_{j,j}$ if $X$ satisfies $C_j$. If clauses $j$ through $k$ are satisfied ($S_{j,k}$) and clauses $k+1$ through $\ell$ are satisfied ($S_{k+1,\ell}$), then in fact clauses $j$ through $\ell$ are satisfied ($S_{j,k} S_{k+1,\ell} = S_{j,\ell}$). We use ``$S$'' for all of these as a mnemonic for ``satisfied'' and also because $S = S_{1,m}$ acts as a ``swap'' on literals.

Note that any generating set for $M$ must include all the $C_j$, since without them, there is no way to generate them from any other elements. Similarly, any generating set must include, for each $i=1,\dotsc,n$, at least one of $X_i$ or $\overline{X}_i$, since again, without them, there is no way to generate those from any other elements.

When $F$ is satisfiable, we claim that $M$ can be generated by precisely $n+m$ elements. Namely, include all $C_j$ in the generating set. Fix a satisfying assignment $\varphi$ to $F$. If $\varphi(X_i)=1$, then include $X_i$ in the generating set, and if $\varphi(X_i)=0$, include $\overline{X}_i$ in the generating set. Since $\varphi$ is a satisfying assignment, for each $j$, there is a literal $X$ in our generating set that appears in $C_j$, and from those two we can generate $C_j X = S_{j,j}$. Next, $S_{j,j} S_{j+1,j+1} = S_{j,j+1}$, and by induction we can generate all $S_{j,k}$. In particular, we can generate $S=S_{1,m}$, and then using $S$ and our literal generators, we can generate the remaining literals.

Conversely, suppose $M$ is generated by $n+m$ elements; we will show that $F$ is satisfiable. As argued above, those elements must consist of $\{C_j : j \in [m]\}$ together with precisely one literal corresponding to each variable. Since the final defining relation can only produce elements $S_{j,\ell}$ with $j$ strictly less than $\ell$, the only way to generate $S_{j,j}$ from our generating set is using the first relation. Namely, at least one of the literals in our generating set must appear in $C_j$. But then, reversing the construction of the previous paragraph, the literals in our generating set give a satisfying assignment to $F$. 

\textbf{Identity elements.} We may add a new element $e$, relations $ea = ae = a$ for all $a$, and include $e$ as a generator. When this is done, the reduction queries whether the algebra is generated by $n+m+1$ elements or more than that many, since the element $e$ must be contained in every generating set.
\end{proof}

As with most $\mathsf{NP}$-complete decision versions of optimization problems, we expect that the \emph{exact} version---given a magma $M$ and an integer $k$, decide whether the minimum generating set has size exactly $k$---is $\mathsf{DP}$-complete, but we leave that as a (minor) open question.

\section{Conclusion} \label{sec:conclusion}
The biggest open question about constant-depth complexity on algebras given by multiplication tables is, in our opinion, still whether or not \algprobm{Group Isomorphism} is in $\mathsf{AC}^0$ in the Cayley table model. 
Our results make salient some more specific, and perhaps more approachable, open questions that we now highlight. 

First, sticking to the generator-enumerator technique, it would be interesting if the complexity of any single part of our Theorem~\ref{thm:QuasigroupIsoPolylogtime} could be improved, even if such an improvement does not improve the complexity of the overall algorithm. Enumerating generators certainly needs $\log^2 n$ bits. Can one verify that a given list of elements is a cube generating sequence better than $\alphacc{1} \betacc{1}\DTISPpll$? Can one verify that a given mapping $g_i \mapsto h_i$ of generators induces an isomorphism with complexity lower than $\alphacc{1}\DTISPpll$? 

More strongly, is \algprobm{GpI} in $\betacc{2}\ACz$? We note that unlike the question of $\ACz$, a positive answer to this question does not entail resolving whether \algprobm{GpI} is in $\P$.

\begin{question}
Does \algprobm{MGS} for groups belong to $\LogSpace$?
\end{question}

\begin{question}
Does \algprobm{Membership} for quasigroups belong to $\LogSpace$?
\end{question}

The analogous result is known for groups, by reducing to the connectivity problem on Cayley graphs. The best known bound for quasigroups is $\SAC^{1}$ due to Wagner \cite{WagnerThesis}. Improvements in this direction would immediately yield improvements in $\algprobm{MGS}$ for quasigroups. Furthermore, a \textit{constructive} membership test would also yield improvements for isomorphism testing of $O(1)$-generated quasigroups. Note that isomorphism testing of $O(1)$-generated groups is known to belong to $\LogSpace$ \cite{TangThesis}.

\begin{question}
What is the complexity of \algprobm{Minimum Generating Set} for semigroups (in the Cayley table model)? More specifically, is it \NP-complete?
\end{question}

\printbibliography

\end{document}